\documentclass[11pt]{article}

\usepackage{graphicx}
\usepackage{verbatim}
\usepackage{fullpage}
\usepackage{hyperref}
\usepackage{amssymb}
\usepackage{amsmath}
\usepackage{amsthm}
\usepackage{amsfonts}
\usepackage{algorithm}
\usepackage{ifthen}
\usepackage[noend]{algorithmic}

\newtheorem{theorem}{Theorem}
\newtheorem{proposition}[theorem]{Proposition}
\newtheorem{lemma}{Lemma}
\newtheorem{claim}[lemma]{Claim}

\newtheorem{problem}{Problem}
\newtheorem{corollary}[lemma]{Corollary}
\newtheorem{example}{Application}

\def\reals{\mathbb{R}}
\def\R{\reals}
\def\sph{\mathbb{S}}

\def\eps{\epsilon}

\def\Lag{\mathcal{L}}
\def\poly{\mathrm{poly}}
\def\HH{{\cal H}}
\def\B{\mathbb{B}}

\newcommand{\abs}[1]{\left|#1\right|}
\newcommand{\norm}[1]{\left\|#1\right\|}

\newcommand{\prob}[1]{{\sf Pr}\left(#1\right)}

\newcommand{\vol}[1]{\operatorname{vol}\left(#1\right)}
\newcommand{\E}[1]{\mathbb{E}\left(#1\right)}
\newcommand{\T}[1]{#1^\mathrm{T}}
\newcommand{\innerprod}[2]{\left<{ #1},{#2}\right>}

\newcommand{\spn}[1]{\mathrm{span}\left(#1\right)}
\newcommand{\EE}[2]{\mathbb{E}_{#1}\left(#2\right)}

\newboolean{longversion}
\setboolean{longversion}{true}

\newboolean{pictures}
\setboolean{pictures}{true}

\title{Structure from Local Optima: Factoring Distributions and Learning Subspace Juntas\\ (Extended version)}
\author{Santosh Vempala
and Ying Xiao\\
School of Computer Science \\
Georgia Institute of Technology \\
{\small\tt \{vempala,ying.xiao\}@gatech.edu}}
\date{\today}

\begin{document}
\maketitle

\begin{abstract}
  Independent Component Analysis (ICA), a well-known approach in
  statistics, assumes that data is generated by applying an affine
  transformation of a fully independent set of random variables, and
  aims to recover the orthogonal basis corresponding to the
  independent random variables. We consider a generalization of ICA,
  wherein the data is generated as an affine transformation applied to
  a product of distributions on two orthogonal subspaces, and the goal
  is to recover the two component subspaces.  Our main result,
  extending the work of Frieze, Jerrum and Kannan, is an algorithm for
  generalized ICA that uses local optima of high moments and recovers
  the component subspaces. When one component is on a $k$-dimensional
  ``relevant" subspace and satisfies some mild assumptions while the
  other is ``noise" modeled as an $(n-k)$-dimensional Gaussian, the
  complexity of the algorithm is $T(k,\eps) + \poly(n)$ where $T$
  depends only on the $k$-dimensional distribution.  We apply this
  result to learning a {\em $k$-subspace junta}, i.e., an unknown
  $0$-$1$ function in $\R^n$ determined by an unknown $k$-dimensional
  subspace. This is a common generalization of learning a $k$-junta in
  $\R^n$ and of learning an intersection of $k$ halfspaces in $\R^n$,
  two important problems in learning theory.

  Our main tools are the use of local optima to recover global
  structure, a gradient-based algorithm for optimization over tensors,
  and an approximate polynomial identity test. Together, they
  significantly extend ICA and the class of $k$-dimensional labeling
  functions that can be learned efficiently.
\end{abstract}
\thispagestyle{empty}
\newpage
\pagenumbering{arabic}

\section{Introduction}\label{sec:introduction}
Independent Component Analysis (ICA) \cite{Jutten1991} is a
statistical approach that models data in $\R^n$ as generated by a
distribution consisting of $n$ linear combinations of $n$ independent
univariate component distributions, $y = Ax$ with $x,y \in \R^n$,
$x_i$ are independent random variables and $A$ is an invertible $n
\times n$ matrix; in other words, an affine transformation of a
product distribution. The goal is to recover the underlying component
distributions of the $x_i$ given only a set of observations
$y$. Special cases of ICA are of interest in many application areas
with large or high-dimensional data sets \cite{ICA01}. An important
feature of ICA, as we will presently see, is that it can provide an
insightful representation even when Principal Component Analysis (PCA)
does not.

In this paper, we consider {\em generalized ICA}, where instead of $n$
independent one-dimensional distributions, we only assume {\em two
  independent distributions on complementary subspaces}. This natural
extension of ICA provides a common generalization of two fundamental
problems in high-dimensional learning, where one sees labeled points
(examples) from an unknown distribution labeled by an unknown $0$-$1$
function and the goal is to find a labeling function that agrees on
most of the distribution \cite{PAC-Valiant}. The first, introduced by
A.~ Blum \cite{Blum94}, is learning a function of $k$ coordinates in
$\R^n$, known as a $k$-junta. The second is the problem of learning an
intersection of $k$ halfspaces in $\R^n$ \cite{BK93, BK97} ($k=1$ is
the classic problem of learning a halfspace). Although the complexity
of both problems is far from settled, there has been much progress in
recent years for special cases, as we discuss in Section
\ref{sec:relatedwork}. Indeed, generalized ICA can be applied to the
problem of learning an unknown function of an unknown $k$-dimensional
subspace of $\R^n$, provided the distribution on points can be
factored into independent distributions on the $k$-dimensional
``relevant" subspace and the $(n-k)$-dimensional ``noise" subspace.

We give an algorithm for generalized ICA that can be viewed as a
tensor version of PCA applied to higher moments, specifically local
optima of moment functions to infer the component distributions. The
algorithm uses a second-order gradient descent method and an
approximate version of the Schwartz-Zippel polynomial identity test,
while its analysis needs tools from convex geometry and probability.
Before we describe our results and techniques in detail, we summarize
the known algorithmic approaches to ICA.

For the problem of identifying the source components given only their
linear combinations as data, PCA suggests the approach of using
principal components of the data as candidates for the component
directions. This would indeed recover the components if the covariance
matrix of the data has distinct nonzero eigenvalues. However, if
variances along two or more directions are equal, then the principal
components are not uniquely defined and PCA does not work. In more
detail, assume that the data is centered, i.e., its mean is zero.
Then PCA can be viewed as finding vectors on the unit sphere that are
local optima of the second moment of the projection, $\max_{x \in
  \sph^{n-1}} \|Ax\|^2$ where $A$ is $m \times n$ with each row being
a data point. These maxima are eigenvalues of $A^TA$, the covariance
matrix of $A$ and hence attain at most $n$ distinct values. The values
and the corresponding vectors can be approximated to arbitrary
accuracy efficiently.

What to do when eigenvalues are repeated? To address this, the idea in
ICA is to consider a broader class of functions to optimize. A natural
choice is higher moments. The use of local optima of fourth moments
was suggested as early as 1991 \cite{lacoumeruiz, Comon91}. When the
component distributions are sufficiently far from being Gaussian, the
local optima of a family of functions on the unit sphere are the
component directions \cite{sw90, dl95} (if component distributions are
Gaussians, then their linear combinations are also Gaussian and the
linear transformation $A$ might not uniquely defined). This approach
can be turned into an a polynomial-time algorithm for unraveling a
product distribution of a wide class of one-dimensional distributions.


We now describe generalized ICA, which significantly weakens
the ICA assumption of a full product distribution. Namely, we assume that the distribution $F$
in $\R^n$ can be factored into a product of two independent marginal
distributions $F_V$ and $F_W$ on unknown orthogonal subspaces $V$ and
$W=V^\perp$, i.e., $F=F_V F_W$. We call such an $F$
\emph{factorizable}. Thus, a random point in $F$ is generated by first
picking its coordinates in $V$ according to $F_V$ and then
independently picking coordinates in $W$ according to $F_W$. The
corresponding problem is the following.
\begin{problem}[Factoring distributions]\label{problem:factoring}
  Given (unlabeled) samples from a factorizable distribution $F=F_V
  F_W$ over $\R^{n}$ (with $V$ and $W$ unknown), recover a
  factorization of $F$.
\end{problem}
If $F$ in fact factorizes further into product of more distributions,
or even a full product distribution of one-dimensional component
distributions as in ICA, an algorithm for the above problem can be
applied recursively to find the full factorization. We will give an
algorithm for this problem under further mild assumptions (roughly
speaking, at least one of $F_V,F_W$ is sufficiently different from
being a Gaussian). Our approach is based on viewing PCA as a second
moment optimization problem, then extending this to higher moments
(alternatively, optimization over tensors). Although such tensor
optimization is intractable in general, for our setting, it will turn
out that {\em local optima} provide valuable information, and can be
approximated efficiently.

The factoring problem above has direct applications to learning in
high dimension. Let $\pi_V$ denote projection to a subspace $V$. We
consider labeling functions $\ell:\R^n \rightarrow \{0,1\}$ 
of the form $\ell(x)
=\ell(\pi_V(x))$. We are given points according to some distribution
$F$ over $\R^n$ along with their labels $\ell(x) = \ell(\pi_V(x))$ for
some unknown subspace $V$ of dimension $k$ (the `relevant' subspace),
and wish to learn the unknown concept $\ell$, i.e., find a function
that agrees with $\ell$ on most of $F$. We call this the problem of
learning a {\em $k$-subspace junta}. We further assume that $F$ is
factorizable as $F=F_VF_W$, with $W = V^{\perp}$ (the `irrelevant'
subspace). The justification for this factorizability assumption is
that coordinates in the $W$ subspace are not relevant to the labeling
function and can be considered to be noisy attributes. The full
statement of our learning problem is as follows:
\begin{problem}[Learning a $k$-subspace junta]\label{problem:learning}
  For $\epsilon, \delta > 0$, given samples drawn from a factorizable
  distribution $F=F_{V}F_{W}$, and labeled by a $\ell = \ell \circ \pi_V$,
  find a $0$-$1$ function $f$ such that with probability at least
  $1-\delta$,
\[
\Pr_F(\ell(x) \neq f(x)) \le \epsilon.
\]
\end{problem}
Our algorithm for generalized ICA leads to an efficient algorithm
for learning $k$-subspace juntas for a large class of ambient
distributions $F$.

\subsection{Related work}\label{sec:relatedwork}
Jutten and Herault formalized the ICA problem
\cite{Jutten1991} and mention in their paper that
variants of this problem had appeared in a variety of different fields
prior to this (the earliest such mention is in \cite{bcs82}). The notion that random variables
should be far from being Gaussian pervades ICA research. By the central
limit theorem, sums of independent random variables converge to a
Gaussian, whereas individually the latent random variables
are not Gaussian. Thus finding directions that maximize some notion of
non-Gaussianity might reveal the latent variables. This intuition is
formalized by introducing functions which serve as a proxy for
non-Gaussianity, called
``contrast functions'' in the ICA literature. 
The definition of a contrast function is that
maximizing a contrast function will give an independent
component. Some examples of contrast functions include the kurtosis
(4th order analogue of variance)\cite{lacoumeruiz, Comon91}, various
cumulants, and functions based on the so-called {\em negentropy}
(\cite{Comon94}). Additionally, there are a variety of tensor methods
and maximum likelihood methods used \cite{cardoso89, cardoso96, Bellsejnowski}.
While there are many algorithms proposed for ICA, some of which appear to perform
well in practice (e.g., FastICA \cite{fastica99}),
there are almost no explicit time complexity bounds. 
Frieze, Jerrum
and Kannan \cite{Frieze96} were the first to give a polynomial
complexity bound for this special case of ICA, namely a product of
uniform distributions on intervals, which can also be viewed as the
problem of learning an unknown parallelopiped from samples.
They used fourth moments, an idea presented earlier in 
several papers in the ICA literature; the key structural lemma is
already present in \cite{dl95}, which was inspired by \cite{sw90}
(Lemma \ref{lemma:representation} of our paper is a
generalization). Subsequently, Nguyen and Regev \cite{nr06} simplified Frieze et al's 
gradient descent algorithm and provided some cryptographic applications.

A different motivation for our work comes
from computational learning theory, where
learning a $k$-{\em junta} is a fundamental problem \cite{Blum94}.
In this problem, one is given points
from some distribution over $\{0,1\}^{n}$, labeled by a
Boolean function that depends only on $k$ of the $n$ coordinates. The
goal is to learn the relevant $k$ coordinates and the labeling
function. Naive enumeration of $k$ subsets of the coordinates leads to
an algorithm of complexity roughly $n^{k}$. Mossel et al \cite{MOS03}
gave an algorithm of complexity roughly $O(n^{0.7k})$ assuming the
uniform distribution over $\{0,1\}^n$.
For other special cases of Problem \ref{problem:learning},
previous authors have applied standard low-dimensional representation
techniques, low-degree polynomials, random projection and Principal
Component Analysis (PCA) to identify $V$ under strong distributional
assumptions \cite{Baum90, KOS08, BK93, Vempala97,Vem10}. The strongest
result in this line achieves a fixed polynomial dependence on $n$ by
applying PCA to learn convex concepts over Gaussian input distributions
\cite{Vempala10}. Unfortunately, standard PCA does not work for other
distributions or more general concept classes, in part because PCA does not
provide useful information when the covariance matrices of the positive and
negative samples are equal.
In fact, the problem appears to be quite hard with no assumptions
on the input distribution, even for small values of $k$, e.g., a single
halfspace can be PAC-learned via linear programming, but learning an
intersection of two halfspaces (a $2$-subspace junta) in polynomial
time is an open problem.

There have been a number of extensions of PCA to tensors
\cite{Kolda09} analogous to SVD, although no method is known to have
polynomial complexity. One approach is to view PCA as an optimization
problem. The top eigenvector is the solution to a matrix optimization
problem:
\begin{align*}
  \max_{\norm{v}=1} v^T A v = \sum_{i_1,i_2} (A)_{i_1,i_2} v_{i_1}v_{i_2}
\end{align*}
where $A$ is the covariance matrix. A higher moment method optimizes
the \emph{multi}linear form defined by the tensors of higher moments:
\begin{align*}
  \max_{\norm{v}=1} A(v,\ldots,v) = \sum_{i_1, \ldots,i_r} A_{i_1,
    \ldots,i_r}v_{i_1}\ldots v_{i_r}.
\end{align*}
Unlike the bilinear case, finding the global maximum of a multilinear
form is hard. For $\alpha > 16/17$, it is NP-hard to approximate the
optimum to better than factor $\alpha^{\lfloor r/4 \rfloor}$
\cite{BruThesis}, and the best known approximation factor is roughly
$n^{r/2}$. Several local search methods have been proposed for this
problem as well\cite{kolda}. 

\subsection{Results}
To state our results formally, we need to define the distance of a
distribution from a Gaussian via moments. For a random vector
$x\in\R^{n}$ with distribution $F$, the $m^{th}$ moment tensor $M^m$
is a tensor of order $m$ with $n^{m}$ entries given by:
\begin{align*}
  M_{i_{1},\ldots,i_{m}}^{m}=\E{x_{i_{1}}\ldots x_{i_{m}}}.
\end{align*}
Let $\Gamma^n$ be the standard Gaussian distribution over $\R^n$ and
$\gamma_m$ denote the $m^{th}$ moment of a standard Gaussian random
variable: $\gamma_m = (m-1)!!$ when $m$ is even and 0 when $m$ is odd.

The {\em $m^{th}$-moment distance} of two distributions $F, G$ over
  $\R^n$ is defined as
\begin{align*}
  d_m(F,G) = \max_{\|u\|=1} \abs{ \EE{F}{ (\T{x} u)^m } - \EE{G}{
      (\T{x} u)^m } } = \norm{M^m_F - M^m_G}_2.
\end{align*}
We say that a distribution $F$ over $\R^k$ is $(m, \eta)$-{\em
  moment-distinguishable} along unit vector $u \in \R^k$, if either
there exists $j \le m$:
\begin{align*}
\abs{\EE{F}{(x^Tu)^j} - \gamma_j} \ge \eta
\end{align*}
or there exist unit vectors $\{v_1,\ldots,v_t\} \subset u^\perp $
where $t \le m$ such that
\begin{align*}
\abs{\EE{F}{(x^Tu)^{m-t}\Pi_{i=1}^t(x^T v_i )} -\EE{F}{(x^Tu)^{m-t}}\E{\Pi_{i=1}^t(x^T v_i )} } \ge \eta.
\end{align*}
In words, $F$ differs from a Gaussian either along some direction $u$,
or by exhibiting a correlation between its marginal along $u$ and vectors orthogonal to $u$ (for
a Gaussian such subsets have zero correlation). The rationale for this
definition is that if two continuous distributions are identical (or
close) in many moments, then one would expect them to be close in
$L^1$ distance. For example, the following holds for one-dimensional
logconcave distributions via an explicit bound on the number of
moments required.
\begin{lemma}[$L^1$ distance from
  Gaussian] \label{theorem:L1gaussian} Fix $m$ and $\eps > 0$. Let
  $f:\R \to \R $ be an isotropic logconcave density, whose first $m$
  moments satisfy $\abs{\EE{f}{x^m} - \gamma^m} < \eps$, then:
  \begin{align*}
  \norm{f-g}_1 \le \left( \frac{c }{m^{1/8}} + c' me^m
    \eps^2\right)^{1/2}\log m \le c\left( \frac{ \log m}{m^{1/16}} +
    \eps m e^m\right)
\end{align*}
\end{lemma}
We are now ready to state our first main result: 
we can efficiently factorize distributions assuming the
distribution on the relevant subspace is moment-distinguishable and
the distribution on the irrelevant noisy attributes is some
Gaussian. In what follows, it might be illustrative to regard $k$ as a
constant independent of $n$. Let $C_F(n,m,\eps)$ be the number of
samples needed to estimate each entry of the $m^{th}$ moment tensor of
$F$ to within additive error $\eps$ and $M$ be an upper bound on the $m^{th}$
moment along any direction.
\begin{theorem}[Factoring, Gaussian noise]\label{thm:factoring-gaussian}
  Let $F=F_V F_W$ be a distribution over $\R^n$ where $V$ is a
  subspace of dimension $k$, and $F_W=\Gamma^{n-k}$. Suppose that
  $F_V$ is $(m,\eta)$-moment-distinguishable for each unit vector $u
  \in V$. Then for any $\eps, \delta \ge 0$, in time $C_F(n,m,\eps)
  \poly(n,\eta,1/\eps, \log(1/\delta),M)$, Algorithm
  \textbf{FactorUnderGaussian} finds a subspace $U$ of dimension at
  most $k$ such that for $j \le m$, $d_j(F,F_UF_{U^\perp}) \le j(M+\gamma_j)\eps$
  with probability at least $1-\delta$. In addition, 
  for any vector in
  $u \in U$, $\|\pi_V(u)\| \ge 1-\eps$. 
\end{theorem}

Next we turn to learning. For a distribution $F$ and a
$k$-dimensional concept class $\HH$, we say that the triple
$(k,F,\HH)$ is \emph{ $(m,\eta)$-moment-learnable} if:
\begin{enumerate}
\item $F=F_V F_W$ is a factorizable distribution with $dim(V)=k$.
\item $\HH$ is a set of $k$-subspace juntas whose relevant subspaces
  are contained in $V$.
\item For $\ell \in \HH$ with minimal (with respect to dimension)
  relevant subspace $P \subseteq V$, for each unit vector $u \in P$,
  either $F_V$ or $F_V^+$ (the distribution over the positive samples)
  is $(m,\eta)$-moment distinguishable along $u$.
\end{enumerate}
In words, the third condition says that if $F_V$ resembles a Gaussian in its 
first $m$ moments along every direction, then $F_V^+$ does not. 
We will see examples of concept classes and distributions for which $m$ is bounded 
 under this definition. Indeed, we conjecture that a
concept class $\HH$ with bounded VC dimension $d$ is $(m,\eta)$
moment-learnable where $m$ depends only on $d$ and $\eta$.

To state our learning guarantee, we need one more definition: A triple
$(k,F,\HH)$ is called \emph{robust} if for any subspace $U$ of
dimension at most $k$ with orthonormal basis $\{u_i\}$ where
$\abs{u_i^T\pi_V(u_i)} \ge 1-\eps$, then $\ell(\pi_U(x))$ labels
correctly $1-g(\eps)$ fraction of $\R^n$ under $F$ where
$g(\eps)<\eps^{c}$ for constant $c>0$ and sufficiently small $\eps$.
The definition requires the distribution $F$ and labeling function
$\ell$ to be robust under small perturbations of the relevant
subspace. Once we identify the relevant subspace approximately, we can
project samples to it and use an algorithm that can learn $\ell$ in
spite of a $g(\eps)$ fraction of noisy labels.
\begin{theorem}[Learning, Gaussian noise]\label{thm:learning-gaussian-noise}
  Let $\eps,\delta >0$, let $\ell \in \HH$ where $(k,F,\HH)$ is $(m,
  \eta)$-moment-learnable and robust, and let $F_W=\Gamma^{n-k}$ be
  Gaussian. Suppose that we are given labeled examples from $F$, then
  Algorithm \textbf{LearnUnderGaussian} identifies a subspace $U$ and
  a hypothesis $h$ such that $h$ correctly classifies $1-\eps$ of $F$
  according to $\ell$ with probability at least $1-\delta$. The time
  and sample complexity of the algorithm are bounded by $T(k,\eps)+
  C_F(n,m,\eps) \poly(n,\eta,k,1/\eps, \log (1/\delta),M)$ where $T$
  is the complexity of learning the $k$-dimensional concept class
  $\HH$.
\end{theorem}
We note here that for a concept class of VC-dimension $d$, a standard
reduction implies that the complexity of learning with $\eps$
arbitrary noise is at most $(2/\eps)^{O(d\log (1/\eps))}$ times the
complexity of learning with no noise (Proposition
\ref{proposition:vc}).
Our algorithms run in polynomial-time in $n$ provided $(k,F,\HH)$
satisfy the moment-learnable condition. Some special cases of this
result were previously known, e.g., when $F$ is a Gaussian and $\HH$
is a convex concept class \cite{KOS08, Vempala10}.
The application of PCA to learning convex bodies
in \cite{Vempala10} can be viewed as the assertion that convex concepts in $\R^k$
are moment-learnable: under a Gaussian distribution, the positive
distribution $F^+$ has variance less than 1 along any direction.
The following two examples further illustrate Theorem \ref{thm:learning-gaussian-noise}.
\begin{itemize}
\item When the full distribution in the
  relevant subspace is uniform in an ellipsoid, then robust concept
  classes can be learned in time $T(k,\eps) + C_{k,\eps}\cdot n^{2}$. Here
  $T$ depends on the $k$ and concept class, and $C$ is a constant
  fixed by $k$ and $\eps$ and independent of the concept class. Thus
  we can learn general concept classes beyond convex bodies and
  low-degree polynomials for uniform distributions over a ball in the
  relevant subspace.
\item When the distribution on the positive examples $F^+$ has bounded support,
i.e., the positive
  labels lie in a ball of radius $r(k)$, such robust concepts can be
  learned in time $T(k,\eps)+C_{k,\eps}\cdot n^{O(r(k)^2)}$ for an arbitrary
  distribution in the relevant subspace. Previously, for logconcave $F$,
  learning an intersection of $k$ half-spaces was known to
  have complexity growing as $n^{O(k)}$ \cite{Vem10, KLT09}.
\end{itemize}

\subsection{Techniques}
Our strategy is to identify the relevant subspace $V$ is to examine
higher moments of the distribution.


As mentioned earlier, our approach is inspired by viewing PCA as
finding the global maxima and minima of the bilinear form defined by
the covariance matrix. Instead of trying to compute \emph{global}
optima of the multilinear form, we use \emph{local} optima. These
local optima turn out to be highly structured. The use of local optima
can be viewed as an effective realization of higher-order PCA that
leads to efficient algorithms. Previous algorithmic problems have all
required the use of global optima --- for example, the planted clique
algorithms of \cite{FK08, Bru09}. We prove that local optima of the
$m^{th}$ moment $f_m(u)=\E{(x^{T}u)^{m}}$ must lie entirely in $V$ or
its complement $W$ (Lemma \ref{lemma:support}) unless its first $m$
moments are identical to those of a Gaussian.

To make these ideas algorithmic, we use a local search method that
increases the function value by performing first-order moves along the
gradient and then second-order moves in the direction of the top
eigenvector of the Hessian matrix. These second-order moves allow us
to avoid saddle points and other critical points which arise in higher
dimensions. Saddle points have a gradient of zero and look like maxima
in some directions and minima in others. While searching for a local
maximum, one could end up in a saddle point. The top eigenvector of
the Hessian shows directions of greatest quadratic increase, and hence
will move us from the saddle point to a true local maximum.

Another component in our algorithms is an approximate version of the
well-known Schwartz-Zippel polynomial identity test. Observing that
$f_m(u)$ is a polynomial of degree $m$ in the variables
$u_1,\ldots,u_m$, in principle we can test ( whether $f_m$ is a
constant function by evaluating $f_m$ at random points. We use a
robust version of this test (Lemma \ref{lemma:sz}) derived via a
result of Carbery and Wright \cite{Carbery}.


\section{Structure of local optima}\label{sec:structure}
We derive a representation of $f_m(u) = \E{(x^tu)^m}$ in Lemma
\ref{lemma:representation}. Using this representation, we show in
Lemma \ref{lemma:support} that each local optimum lies in $V$ or $W$
exclusively. Finding a sequence of orthogonal local optima will
give us basis vectors for the relevant subspace.

For convenience we often use $u_{V}=\pi_V(u)$ for the projection of
$u$ onto $V$, $u_{W}$ for the projection onto the orthogonal subspace
$W$, and $u^{0}$ for the unit vector in the direction of $u$.

We may assume that $\E{x}=0$: if otherwise, then we can apply a
translation $x-\E{x}$.
\ifthenelse{\boolean{longversion}}{
\begin{lemma}[Translation of product distributions]\label{lemma:translation}
  Let $x\in \R^n$ be a random vector drawn from $F=F_VF_W$, a product
  distribution.Then $x-\E{x}$ has a product distribution over $V$ and
  $W$.
\end{lemma}
\begin{proof}[Proof of Lemma \ref{lemma:translation}]
Take our translation  $y=T_a(x)=x+a$, for Borel sets $B_1$ and $B_2$:
\begin{align*}
  \prob{y_V \in B_1 \land y_W \in B_2} & = \prob{x_V + a_V \in B_1 \land
    x_W + a_W \in B_2} \\
  & = \prob{x_V \in B_1 - a_V \land x_W
    \in B_2 - a_W} \\
  & = \prob{x_V \in B_1 - a_V}\prob{x_W \in  B_2 - a_W} \\
  & = \prob{y_V \in B_1} \prob{y_W \in B_2}
\end{align*}
\end{proof}
} 
We can combine this with a linear transformation to obtain an
isotropic distribution, given by $y=\Sigma^{-1/2}(x-\mu)$ where $\mu$
is the expectation vector. This simplifies subsequent calculations
because the covariance matrix for $y$ is $I_n$. The following lemma, 
inspired by \cite{Frieze96,dl95,sw90}, 
provides the main insight for the structural theorem. 
\begin{lemma}[Representation of $f_{m}$]\label{lemma:representation}
  Let $F=F_VF_W$. Suppose that $x$ has the same $j^{th}$ moments as a
  Gaussian for all integers $j<m$, then for $u \in \sph^{n-1}$:
  \begin{align}
    f_{m}(u)=\norm{u_{V}}^{m}\left(\E{(\T{x_{V}}u_{V}^{0})^{m}}-\gamma_m\right)+
    \norm{u_{W}}^{m}\left(\E{(\T{x_{W}}u_{W}^{0})^{m}}-\gamma_m\right)+\gamma_m
  \end{align}
\end{lemma} 
\begin{proof}[Proof of Lemma \ref{lemma:representation}]
  Consider the case when $m$ is odd:
\begin{align*}
  f_m(u) & = \E{(x^Tu)^m}\\
  & =\E{\left(\T{(x_{V}+x_{W})} (u_{V}+u_{W})\right)^{m}}\\
  & = \E{(x_V^T u_V)^m}+\E{(x_W^T u_W)^m}+\sum_{i=1}^{m-1}\binom{m}{i}
  \E{(x_{V}^{T}u_{V})^{i}(x_{W}^{T}u_{W})^{m-i}} \\
  & = \E{(x_V^T u_V)^m}+\E{(x_W^T u_W)^m}+\sum_{i=1}^{m-1}\binom{m}{i}
  \E{(x_{V}^{T}u_{V})^{i}}\E{(x_{W}^{T}u_{W})^{m-i}}
\end{align*}
The last line follows by applying the independence of random variables
which depend only on the $V$ and $W$ subspaces. Each term in the last
sum contains an odd moment of a Gaussian, hence:
\begin{align}
  f_{m}(u)=\norm{u_{V}}^{m}\E{(\T{x_{V}}u_{V}^{0})^{m}}+
  \norm{u_{W}}^{m}\E{(\T{x_{W}}u_{W}^{0})^{m}}.
\end{align}
When $m$ is even, we need the following formula:
\begin{align*}
  \sum_{i=0}^{m}\binom{m}{i}\norm{u_{V}}^{i}\norm{u_{W}}^{m-i}
  \gamma_i \gamma_{m-i}=\gamma_m
\end{align*}
This follows from $\E{(aX+bY)^m}=\gamma_m$ where $a^2+b^2=1$ and
$X$ and $Y$ are independent standard normal variables:
\begin{align*}
  f_{m}(u) 
  &=  \sum_{i=0}^{m}\binom{m}{i}\norm{u_{V}}^{i}\norm{u_{W}}^{m-i}
  \E{(x_{V}^{T}u_{V}^{0})^{i}}\E{(x_{W}^{T}u_{W}^{0})^{m-i}} \\
  &= \norm{u_V}^m \E{ (\T{x} u_V^0 )^m} + \norm{u_W}^m \E{ (\T{x} u_V^0 )^m}
+ \sum_{i=1}^{m-1}\binom{m}{i}\norm{u_{V}}^{i} \norm{u_{W}}^{m-i}
  \gamma_{i} \gamma_{m-i} \\
  &= \norm{u_V}^m \left( \E{ (\T{x} u_V^0 )^m} - \gamma_m \right) 
  + \norm{u_W}^m \left( \E{ (\T{x} u_V^0 )^m} - \gamma_m \right)
+ \sum_{i=0}^{m}\binom{m}{i}\norm{u_{V}}^{i} \norm{u_{W}}^{m-i}
  \gamma_i \gamma_{m-i} \\
  &= \norm{u_V}^m \left( \E{ (\T{x} u_V^0 )^m} - \gamma_m \right) 
  + \norm{u_W}^m \left( \E{ (\T{x} u_V^0 )^m} - \gamma_m \right)
+ \gamma_m
\end{align*}
\end{proof} 
Using this representation, we can characterize all local optima of
$f_m$.
\begin{lemma}[Support]\label{lemma:support}
  Let distribution $F=F_VF_W$ have the same first $m-1$ moments as a
  Gaussian but a different $m^{th}$ moment. Then for a local maximum
  (local minimum) $u^\ast$ of $f_{m}$ restricted to the unit sphere,
  where $f_m(u^\ast) > \gamma_m$ ($f_m(u^\ast) < \gamma_m$), either
  $\norm{u_{V}^\ast}=1$ or $\norm{u_{W}^\ast}=1$.
\end{lemma}
\ifthenelse{\boolean{longversion}}{ 
\begin{proof}[Proof of Lemma \ref{lemma:support}]
  Consider the curve $C=\{s (u^\ast_V)^0 + t (u^\ast_W)^0: s^2+t^2=1, s\ge 0,t \ge
  0\}$. The point $u^\ast$ lies on $C$: thus if $u^\ast$ is a local
  maximum in full space, it had better be a local maximum on $C$. On
  the other hand, we will show that there are no local maxima interior
  to $C$, whence we must have $\norm{u_{V}^\ast}=1$ or
  $\norm{u_{W}^\ast}=1$. 

  Let us denote $a_v = \left( \E{ (\T{x} u_V^0 )^m} - \gamma_m
  \right)$ and $a_w = \left( \E{ (\T{x} u_w^0 )^m} - \gamma_m
  \right)$. By the assumption that $f_m(u^\ast) > \gamma_m$, we know
  that least one of $a_v$ or $a_w$ is positive.  Suppose that $s \ne
  0$ and $s \ne 1$: we form the associated Lagrangian with positive
  real multiplier $\lambda$:
  \begin{align*}
    \Lag = a_v s^m + a_w t^m + \gamma_m - \lambda ( s^2 + t^2 - 1)
\end{align*}
At every critical point in the interior of $C$, we must have $D \Lag = 0$:
\begin{align*}
\left(
\begin{array}{l}
m a_v s^{m-1} - 2 \lambda s \\
m a_w t^{m-1} - 2 \lambda t
\end{array}
\right) = 0
\end{align*}
If we consider only the interior critical points where $s,t > 0$, then
both $a_v > 0$ and $a_w > 0$ (otherwise we would have $\lambda > 0$
and $\lambda \le 0$). There is only one solution:
\begin{align*}
s =  \frac{a_w^{1/(m-2)}}{\sqrt{a_v^{2/(m-2)}+a_w^{2/(m-2)}}}  \qquad
t =  \frac{a_v^{1/(m-2)}}{\sqrt{a_v^{2/(m-2)}+a_w^{2/(m-2)}}} \qquad
\lambda = (m/2) (a_v s^m + a_w t^m)
\end{align*}
If we now consider the Hessian on the tangent plane orthogonal to the
gradient of the constraint (equivalent to considering the bordered
Hessian) , we see that it is positive definite for $m>3$ (when $m=3$,
differentiating $a_v s^3 + a_w (1-s^2)^{1.5}$ twice at the critical
point gives a positive value):
\begin{align*}
D^2 \Lag = \left( 
  \begin{array}{ll} 
    m(m-1) a_v s^{m-2} - 2 \lambda& 0 \\
    0 & m(m-1) a_w t^{m-2} - 2 \lambda \\
  \end{array}
\right) = \frac{m(m-3) a_v a_w}{\left[a_v^{2/(m-2)}+a_w^{2/(m-2)}\right]^{(m-2)/2}} I
> 0
\end{align*}
In particular, there are no local maxima interior to $C$, that is,
$\norm{u_V^\ast}=1$ or $\norm{u_V^\ast}=0$.
\end{proof}

}{}

\section{Finding a basis}\label{sec:basis}
Our two basic algorithms exploit the property that local
optimum to $f_m(u)=\E{(x^Tu)^m}$ on the unit sphere must lie in either
$V$ or $W$ (Lemma \ref{lemma:support}). 
In this section, we assume that the algorithms have access
to exact moment tensors and can compute exact local optima. 
We provide
efficient algorithms (with error analysis) in Section
\ref{sec:gaussian}. 

The basic idea of the algorithm is to start with a random direction
and evaluate the $j$'th moment in that direction. If it is different
from a Gaussian we go to a local max or local min (whichever keeps it
different from a Gaussian) and thus find a vector of interest; if many
random unit vectors have Gaussian moments, then all directions have
Gaussian moments due to the Schwartz-Zippel Lemma and we go to the
next higher moment. At the end of the algorithm we have a subset of an
orthogonal basis consistent with $V$ and $W$, and the property that
all orthogonal directions have Gaussian moments.

\begin{lemma}[Schwartz-Zippel\cite{Schwartz}]\label{lemma:schwartz}
  Let $P \in F[x_1,\ldots,x_n]$ be a nonzero polynomial of degree $d
n  \ge 0$ over field $F$. Let $S$ be a finite subset of $F$ and let
  $r_1,\ldots,r_n$ be selected randomly from $S$. Then:
  \begin{align*}
  \prob{P(r_1,\ldots,r_n)=0}\leq\frac{d}{|S|}.
  \end{align*}
\end{lemma}
 
\ifthenelse{\boolean{pictures}}{
\begin{algorithm}
  \caption{FindBasis}
  \begin{algorithmic}[1]
    \REQUIRE{Moment bound $m$, Distribution $F$}
    \STATE Orthonormal vectors $B\leftarrow\phi$, moment tensor $j \leftarrow2$.
    \WHILE{$\abs{B}<n$ and $ j \le m$}

    \STATE Compute the $j^{th}$ moment tensor $M_j^B$ orthogonal to $B$, so that for any $v \in B^\perp$, 
 $f_{j}(v)=\E{(\T{x}v)^{j}} = M_j^B(v,\ldots,v)$.
    
    \IF{$f_j(v/\norm{v}) \equiv \gamma_j $}

    \STATE $j\leftarrow j+1$

    \ELSE
    \IF{$f_{j}(v)>\gamma_j$ for some $v$} 

    \STATE Compute a local maximum $u^{\ast}$ to $f_{j}$ starting from $v$.

    \ELSE 

    \STATE Compute a local minimum $u^{\ast}$ to $f_{j}$ starting from $v$.
    \ENDIF
    \STATE $B\leftarrow B\cup\{u^{\ast}\}$.
    \ENDIF

   \ENDWHILE
   \RETURN $B$
  \end{algorithmic}
\end{algorithm}
}{}

For Line 3, let $A:\R^n \to \R^n$ denote the linear map that projects
orthogonal to $B$. Then
  \begin{align*}
    M(Au,\ldots,Au) & = \sum_{i_1,\ldots,i_m} M_{i_1,\ldots,i_m}
    (Au)_{i_1}\cdots (Au)_{i_m} & = \sum_{j_1,\ldots,j_m} \left(
      \sum_{i_1,\ldots,i_m} M_{i_1,\ldots,i_m} A_{i_1,j_1}\cdots
      A_{i_m,j_m}\right) u_{j_1} \cdots u_{j_m}
  \end{align*}
  The identity check in Line 4 is performed by selecting a random
  vector $x$ with i.i.d. uniform coordinates from $\{1,\ldots,2m\}$
  and evaluating the polynomial $f_j(x/\|x\|) -\gamma_j$. Repeating
  $O(\log(n/\delta))$ times gives a $1-\delta$ probability of success.

  Algorithm \textbf{FindBasis} does not suffice on its own. Although
  every direction orthogonal to $B$ looks Gaussian up to the $m$'th
  moment, it is possible that some directions are correlated with
  vectors in $B$. The next procedure identifies such directions.

\ifthenelse{\boolean{pictures}}{
\begin{algorithm}
  \caption{ExtendBasis}
  \begin{algorithmic}[1]
    \REQUIRE Moment bound $m$, distribution $F$, candidate vectors
    $S$ and non-Gaussian directions $T$.

    \STATE $S' \leftarrow S$, $j \leftarrow 2$.

    \WHILE{$\abs{S'} < k$ and $j \le m$}

    \FOR{each choice (with repetition) $\{v_1,\ldots,v_l\} \subseteq S'$
     where $1 \le l < j$. }

   \STATE Compute the $(j-l)$ tensor $M^{S',T}_{l,j}$ so that for any $u \in (S' \cap T)^\perp$,
\[
g(u) = \E{(\T{x}u)^{j-l}\prod
     (\T{x}v_t)} - \E{(\T{x}u)^{j-l}}\E{\prod (\T{x}v_t)} = M^{S',T}_{l,j}(u,..,u,v_1,\ldots,v_l). 
\]
  
   \IF{$g(u) \equiv 0$}

   \STATE Continue with next choice of $\{ v_i \}$.

   \ELSE

   \IF{$g(u) > 0$ for any $u$}
   \STATE Compute a local maximum $u^\ast$ 
   to $g$ starting with $u/\norm{u}$.
   \ELSE
   \STATE Compute a local minimum $u^\ast$ 
   to $g$ starting with $u/\norm{u}$.
   \ENDIF

   \STATE $S' \leftarrow S' \cup \{u^\ast\}$ and restart the while loop
   with $j=3$.

   \ENDIF
   
   \ENDFOR
   \STATE $ j \leftarrow j+1$.
   \ENDWHILE

   \RETURN $S'$.
  \end{algorithmic}
\end{algorithm}
}{}

\begin{theorem}[Find Basis]\label{thm:exact}
  Let $F=F_V F_W$ be a factorizable distribution over
  $\R^{n}$. 
Then, with probability at least
  $1-\delta$, each vector in the output of \textbf{FindBasis} lies in
  either $V$ or $W$.
\end{theorem}
\begin{proof}[Proof of Theorem \ref{thm:exact}] From the above
  comment, at each step, with probability at least $1-\delta/n$ (hence
  total failure probability $\delta$), we are able to find a point $u$
  where $f_r(u) \ne \gamma_m$. In particular, if $f_r(u) > \gamma_m$,
  then we find a local maximum $u^\ast$. By Lemma \ref{lemma:support},
  $u^\ast$ is contained entirely within $V$ or $W$. The analysis is
  identical when our initial point $u$ satisfies $f_r(u) < \gamma_m$.

  Observe that $F_{V\setminus\spn{B}}F_{W\setminus\spn{B}}$ is a factorizable
  distribution over $\pi_{B^{\perp}}(x)$, and hence a local optimum in
  $B^\perp$ also will lie in either $V$ or $W$. 
\end{proof}

The next theorem states that \textbf{ExtendBasis} finds all
vectors which are correlated with $S \subseteq V$, and that all remaining
vectors at the end of the algorithm are uncorrelated up to the
$m^{th}$ moment.
\begin{theorem}[Basis Extension]\label{thm:extend}
  The output $S'$ of \textbf{ExtendBasis} on input $S \subseteq V, T \subseteq W$ satisfies:
  \begin{enumerate}
    \item $S \subseteq S' \subseteq V$.
    \item For $ \{ v_t \} \subset S'$ and $ \{ u_i \} \subset
      (S')^\perp$:
\begin{align*}
  \E{\prod_{i=1}^{j-l} (\T{x}u_i) \prod_{t=1}^l (\T{x}v_t)} =
  \E{\prod_{i=1}^{j-l} (\T{x}u_i)} \E{\prod_{t=1}^l (\T{x}v_t)}
\end{align*}
\end{enumerate}
\end{theorem}

\ifthenelse{\boolean{longversion}}{
\begin{proof}[Proof of Theorem \ref{thm:extend}]
  The Schwartz-Zippel lemma returns a correct decision at every
  iteration (there are at most $n^k$ of these, so if we pick our
  domain to be of size $2n^k$ and run $O( \log{n^k/\delta})$
  iterations each time, then we have a correct decision for all
  iterations with probability at least $1-\delta$.

  Let $u^\ast$ be a local maximum found by \textbf{ExtendBasis} using
  the $j^{th}$ moment. Consider the $\{v_1,\ldots,v_l\}$ where
  $u^\ast$ was found.
\begin{align*}
  g(u) &= \E{(\T{x}u)^{(j-l)}\prod_{t=1}^l (\T{x}
    v_t)}-\E{(\T{x}u)^{(j-l)}}\E{
    \prod_{t=1}^l (\T{x} v_t)} \\
  & = \sum_{i=0}^{j-l} \binom{j-l}{i} \E{(\T{x}u_W)^i} \left[
    \E{(\T{x}u_V)^{(j-l-i)}\prod_{t=1}^l (\T{x} v_t)} - \E{\prod_{t=1}^{l}
      (\T{x}v_t)} \E{(\T{x}u_V)^{(j-l-i)}}\right]
\end{align*}
Since $u^\ast$ was found at moment $j$, then for all $0<i<j-l$:
\begin{align*}
  \E{(\T{x}u_V)^{(j-l-i)}\prod_{t=1}^l (\T{x} v_t)} =
  \E{\prod_{t=1}^{l} (\T{x}v_t)} \E{(\T{x}u_V)^{(j-l-i)}}
\end{align*}
Only the first and last terms survive:
\begin{align*}
g(u) &= \E{(\T{x}u_V)^{(j-l)}\prod_{t=1}^l (\T{x} v_t)}  - \E{\prod_{t=1}^l (\T{x} v_t)}  \E{(\T{x}u_V)^{(j-l)}}\\
 & = \norm{u_V}^{j-l} \left[ \E{(\T{x}u_V^0)^{(j-l)}\prod_{t=1}^l (\T{x} v_t)}  - \E{\prod_{t=1}^l (\T{x} v_t)}  \E{(\T{x}u_V^0)^{(j-l)}} \right]
\end{align*}
Having a positive local maximum implies that $\norm{u_V}=1$.

For the second part of this lemma: we already know that all the
remaining vectors must have Gaussian moments. Fix $j \le m$ and a
choice of $\{v_1,\ldots v_l\}$ from $S'$ and consider the symmetric
tensor $\hat{M}$ represented by $f(u) -
\E{(\T{x}u)^{j-l}}\E{\prod_{t=1}^{l} (\T{x}v_t)}$. We require the
following claim for symmetric tensors where for any permutation
$\sigma:[m] \to [m]$:
\begin{align*}
  \E{\Pi_{k=1}^{m} x_{i_k}} =\E{\Pi_{k=1}^{m} x_{i_{\sigma(k)}}}.
\end{align*}
\begin{claim}\label{claim:spectral}
If $A$ is a symmetric order $r$ tensor, then:
\begin{align*}
\max_{\norm{v}=1} A(v,\ldots,v) =
\max_{\norm{v_1}=1,\ldots,\norm{v_r}=1} A(v_1,\ldots,v_r)
\end{align*}
\end{claim}
Using Claim \ref{claim:spectral}:
\begin{align*}
\max_{\norm{u}=1} \hat{M}(u,\ldots, u) =
  \max_{\norm{u_1}=1,\ldots,\norm{u_{j-l}}=1} \hat{M}(u_1,\ldots,u_{j-l})
\end{align*}
In particular, there exists $\{u_i\}$ such that
\begin{align*}
  \E{\prod_{i=1}^{j-l} (\T{x}u_i) \prod_{t=1}^l (\T{x}v_t)} > \E{\prod_{i=1}^{j-l} (\T{x}u_i)}
  \E{\prod_{t=1}^l (\T{x}v_t)}
\end{align*}
if and only if there exists $u$ such that
\begin{align*}
\E{(\T{x}u)^{j-l} \prod_{t=1}^l (\T{x} v_t)} > \E{(\T{x}u)^{j-l}} \E{\prod_{t=1}^l (\T{x}v_t)}
\end{align*}
But at the end of the algorithm, we know that there are no such $u$,
hence there can be no such $u_i$ either. Thus, we can factor any $u
\notin S'$ through the expectations which contain only $v_i$ from
$S'$.
\end{proof}
}{}

\subsection{Generalized ICA}\label{sec:general}
We can now give an algorithm for generalized ICA, assuming access to
exact moment tensors and local optima. If $F$ is factorizable,
Algorithm \textbf{Factor} will provide a factoring into subspaces such that 
the marginal distributions look independent up to $m$
moments. 
The output of \textbf{FindBasis} is a set of vectors that each lie in $V$ or in $W$. We try all possibilities for the subset from $V$, 
then extend this using \textbf{ExtendBasis}, consider the resulting decomposition and take the option that gives a 
product factorization. The factorization found will be $U, U^\perp$ for some $U \subseteq V$.
 
\ifthenelse{\boolean{pictures}}{
\begin{algorithm}
\caption{Factor}

\begin{algorithmic}[1]
  \REQUIRE Highest moment $m$, distribution $F$.
  \STATE $B\leftarrow\boldsymbol{FindBasis(m,F)}$. 
  \FOR{every subset $T\subseteq B$ of at most $k$ vectors}
  \STATE $S' \leftarrow \boldsymbol{ExtendBasis(m,F,T,B-T)}$.
  
  \STATE $T' \leftarrow\boldsymbol{ExtendBasis(m,F,B-T,S')}$.

    \IF{$\abs{S'} > k$ or $\abs{T'} > n-k$}
  \STATE Continue with the next choice of $T$.
  \ENDIF

  \STATE Augment $S'$ with $k-\abs{S'}$ orthonormal vectors from
  $\R^n-\spn{T'}$, forming basis $U$.

  \STATE Compute $m$ moments of $F$, $F_U$ and $F_{U^\perp}$:

  \FOR{$l \le m$}
  \STATE Compute:
  \begin{align*}
    \Delta_{U}^l=\sum_{(i_{1},\ldots,i_{l})}\left( \EE{F}{x_{i_1}\cdots
        x_{i_l}}- \EE{F_U}{x_{p_1}\cdots
        x_{p_j}}\EE{F_{U^\perp}}{x_{p_{j+1}}\cdots x_{p_l}} \right)^2
  \end{align*}
  where $\{x_{p_1},\ldots,x_{p_k}\}$ correspond to coordinates in $U$,
  and $\{x_{p_{k+1}},\ldots,x_{p_l}\}$ correspond to coordinates in
  the $U^\perp$ subspace.
  \ENDFOR
  \ENDFOR
  \RETURN $U$ with lowest $\Delta_{U}^3$, breaking ties by considering
  $\Delta_{U}^4,\Delta_{U}^5,\ldots$
\end{algorithmic}
\end{algorithm}
}{}
\begin{theorem}[Factoring, general noise]\label{thm:factor} 
  For any $\eps,\delta>0$, given the moment tensors of distribution
  $F$ over $\R^{n}$ and the ability to compute exact local optima, if
  there exists a subspace $V$ with $dim(V)=k$ such that for
  $j=1,\ldots, m$ $d_j(F,F_V F_W)=0$, Algorithm \textbf{Factor} finds
  a subspace $U$ such that $d(F,F_U F_{U^\perp}) =0$ with probability
  at least $1-\delta$.  The time and sample complexity of the
  algorithm are bounded by $n^{O(k+m)}$.
\end{theorem}
\ifthenelse{\boolean{longversion}}{
\begin{proof}[Proof of Theorem \ref{thm:factor}] In the enumeration of
  all subsets of size at most $k$ subsets of $B$ at line 2, we
  encounter $T = B \cap V$. By Theorem \ref{thm:extend}, the output
  $S'$ of \textbf{ExtendBasis} contains only vectors in $V$ and $T'$
  contains only vectors from $W$. By Part 2 of Theorem
  \ref{thm:extend}, the following holds for any choice of vector $\{
  u_i \} \subset \sph^{n-1}-\spn{S',T'}$, we have $\E{(x^t u_i)^j} =
  \gamma_m$ for $j \le m$ and that:
  \begin{align*}
  \E{\prod_{i=1}^{j-l} (\T{x}u_i) \prod_{t=1}^l (\T{x}v_t)} =
  \E{\prod_{i=1}^{j-l} (\T{x}u_i)} \E{\prod_{t=1}^l (\T{x}v_t)}
\end{align*}
for $v_t \in S' \cup T'$. In particular, every such $u$ is independent
from $S'$ and $T'$ up to the $m^{th}$ moment. In the augmented basis,
the expectations separate into the products over the two subspaces:
  \begin{align*}
    \E{x_{i_{1}}\ldots x_{i_{l}}}=\E{x_{p_{1}}\ldots
      x_{p_{j}}}\E{x_{p_{j+1}}\ldots x_{p_{l}}}
  \end{align*}
  where $\{x_{p_1},\ldots,x_{p_j}\}$ correspond to coordinates in the
  $U$ subspace, and $\{x_{p_{j+1}},\ldots,x_{p_l}\}$ correspond to
  coordinates in the $U^\perp$ subspace. In particular, the entries of
  the moment tensor of $F$ are equal to the entries of the moment
  tensor of $F_U F_{U^\perp}$, and hence will return $\Delta^j=0$.
 \end{proof}
}{}

\section{Gaussian noise model}\label{sec:gaussian}

We now give a complete algorithm assuming $F_W$ is a Gaussian, assuming we only have access
to $F$ through samples (not exact moment tensors). 
The main difficulty is handling the error accumulation
over multiple iterations, as in each round we can only hope to approximately compute moments and find approximate
local optima.
The idea is that \textbf{FindBasis} and \textbf{ExtendBasis} find
vectors where $\E{(x^T u)^m} \neq \gamma_m$. If $F_W$ is Gaussian, our
algorithms only find directions in $V$. Thus, the error accumulates
over only $k$ steps, and the total error depends on $k$ rather than
$n$.

\subsection{Local search}\label{sec:localsearch}
To compute approximate local optima, we perform gradient ascent,
moving in the direction of the gradient. If moving along the gradient
does not increase the function value by a certain value, we switch to
second-order moves based on the Hessian.  We will use the notation
that $Dg_u$ is the gradient of $g$ at $u$ and $D^2g_u$ for the Hessian matrix.  The
top eigenvalue of a matrix on a subspace orthogonal to a vector can be
computed via a coordinate transformation.  We denote $M
=\norm{M^m}_2$.  \ifthenelse{\boolean{pictures}}{
\begin{algorithm}
  \caption{LocalOpt}
  \begin{algorithmic}[1]
    \REQUIRE Function $g$, error parameter $\eps_1$,
    \STATE $u \leftarrow$ uniformly at random over unit sphere.
    \WHILE{ $\abs{\innerprod{u}{Dg_{u} }} \le (1 - \eps_1)\norm{Dg_{u}}
      $ or $\lambda_{\max}(D^2g_{u}) \ge \eps_2$ on $u^\perp$ }
    \IF{ $\abs{\innerprod{u}{Dg_{u} }} \le (1 - \eps_1)\norm{Dg_{u}} $}
    \STATE Direction $v \leftarrow \pi_{u^\perp}(Dg_{u})$.
    \STATE $u \leftarrow u + r_1 v/\norm{v}$.
    \STATE Renormalize $u \leftarrow u /\norm{u}$.
    \ELSIF{$\lambda_{\max}(D^2g_{u}) \ge \eps_2$ on $u^\perp$}
    \STATE Direction $v \leftarrow$ top eigenvector of $D^2g_{u}$ on $u^\perp$.
    \STATE $u \leftarrow u - r_2 v/\norm{v}$.
    \STATE Renormalize $u \leftarrow u /\norm{u}$.
    \ENDIF
    \ENDWHILE
    \RETURN $u$.
  \end{algorithmic}
\end{algorithm}
}{} \textbf{LocalOpt} terminates in polynomial time when the
parameters $\eps_1$, $r_1$, $\eps_2$ and $r_2$, the thresholds and
step sizes for the first-order moves and second-order moves are chosen
appropriately. Note that $\eps_2$ varies with the function value, but
the remaining parameters are fixed.
\begin{lemma}[Local search termination]\label{lemma:gradient}
  Let $g(u)$ satisfy $g(tu)=t^m g(u)$ for some integer $m$. Suppose
  that for our starting point $u$ that $g(u) \ge \eta > 0$. Choose the
  parameters as follows:
  \begin{align*}
    \eps_1 &\le \left( \frac{81m(m-1)^2 \eta^2}{1048M} \right)^2 \qquad
    r_1 \le \frac{\sqrt{\eps_1}}{4m^2M}\\
    \eps_2 &= \frac{3 m(m-1)g(u)}{4} \qquad r_2 \le
    \frac{9\eta}{256(m-2)M}
  \end{align*}
  where $M$ is an upper bound for $g$ on the unit sphere. Then
  \textbf{LocalOpt} will terminate in at most
  $\poly(M,m,1/\eps_1,1/r_1,1/r_2, 1/\eta)$ iterations.
\end{lemma}
\ifthenelse{\boolean{longversion}}{
\begin{proof}[Proof of Lemma \ref{lemma:gradient}]
  Consider an iteration of the algorithm where the first derivative
  condition is unsatisfied, and we make a step of size $r_1$ in the
  direction of $v/\norm{v}$ (call the step $h$). The new function
  value at this point $u+h$ is given by the Taylor series expansion
  with error (where $\zeta$ lies between $u$ and $u+h$):
  \begin{align*}
    g(u+h) & = g(u) + Dg_u \cdot h + \frac{1}{2}h^T( D^2g_\zeta) h
  \end{align*}
  The increase in function value is lower bounded as follows:
  \begin{align*}
    Dg_u \cdot h + \frac{1}{2}h^T D^2g_\zeta h & \ge r_1
    \innerprod{Dg_u}{\frac{v}{\norm{v}}} - \frac{1}{2}r_1^2
    (v/\norm{v})^T D^2g_\zeta (v/\norm{v}) \\
    & \ge r_1 \sqrt{\eps_1} - \frac{1}{2}r_1^2 m^2 M \\
    & \ge r_1 \sqrt{\eps_1} - \frac{1}{8} r_1 \sqrt{\eps_1} \\
    & \ge \frac{7}{8} r_1 \sqrt{\eps_1}
  \end{align*}
  Thus, we have lower bounded the increase in the function value. When
  we rescale $u+h$ back to norm 1, we can apply the $m$-homogeneity of
  $f$ to deduce that:
  \begin{align*}
    g\left( \frac{u+h}{\norm{u+h}} \right) =
    \frac{1}{\norm{u+h}^{m/2}} g(u+h)
  \end{align*}
  We can compute $\norm{u+h} = 1 + r_1^2$ because $r_1$ is
  perpendicular to $u$. Hence:
  \begin{align*}
    g\left( \frac{u+h}{\norm{u+h}} \right) & =
    \frac{1}{(1+r_1^2)^{m/2}} g( u + h ) \\
    & \ge \left(1-\frac{m+2}{2} r_1^2\right) g(u+h) \\
    & \ge g(u) \left( 1 + \frac{7}{8g(u)} r_1 \sqrt{\eps_1}\right)
    \left( 1- \frac{m+2}{2} r_1^2 \right)
  \end{align*}
  where we used the estimate:
  \begin{align*}
    (1+x)^k \ge 1 + ( k + 1/2) x 
  \end{align*}
  for $x \le 2/k^2$. To finish this calculation, we simply substitute
  our value for $r_1$ in terms of $\eps_1$:
  \begin{align*}
    g\left( \frac{u+h}{\norm{u+h}} \right) & \ge g(u) \left( 1 +
      \frac{7}{8g(u)} r_1 \sqrt{\eps_1} \right) \left( 1-
      \frac{1}{8M} r_1 \sqrt{\eps_1}  \right) \\
    & \ge g(u) \left( 1 + \frac{5}{8M} r_1 \sqrt{\eps_1} \right)
  \end{align*}
  Hence, there are at most at most a polynomial number of iterations
  of this form.  Consider now an iteration where the second derivative
  condition is unsatisfied (and the first derivative condition must be
  satisfied). We take the same Taylor series expansion with error term
  (to one further term), where $h = r_2 v / \norm{v}$:
  \begin{align*}
    g(u+h) & = g(u) + Dg_u \cdot h + \frac{1}{2}h^T D^2g_u h +
    \frac{1}{6} D^3g_\zeta (h,h,h)
  \end{align*}
  We will show that the contributions from the first and third
  derivative terms are small, and that the second derivative term
  dominates. In the first derivative term, note that $h$ is orthogonal
  to $u$, and hence the component of $Dg_u$ parallel to $h$ has norm
  at most $\sqrt{2\eps_1-\eps_1^2} \norm{Dg_u}$. We estimate the other
  terms as before:
  \begin{align*}
    Dg_u \cdot h + \frac{1}{2}h^T D^2g_u h + \frac{1}{6} D^3g_\zeta
    (h,h,h) & \ge -\sqrt{2\eps_1-\eps^2}mM + \frac{1}{2}r_2^2 \eps_2 -
    \frac{m(m-1)(m-2)}{6}r_2^3 M \\
    & \ge -\frac{1}{128}r_2^2\eps_2 + \frac{1}{2}r_2^2 \eps_2 -
    \frac{1}{128}r_2^2 \eps_2 \\
    & \ge \frac{31}{64} r_2^2 \eps_2
  \end{align*}
  Once again, we have to rescale back to norm 1. In this case:
  \begin{align*}
    g\left( \frac{u+h}{\norm{u+h}} \right) & \ge g(u) \left( 1 +
      \frac{31}{64g(u)} r_2^2 \eps_2 - \frac{m+1}{2} r_2^2 \right) \\
     &  \ge g(u) \left( 1 + \frac{93}{256} m(m-1) r_2^2 - \frac{m+1}{2}
       r_2^2 \right) \\
     & \ge g(u) \left( 1 + \frac{93}{256}r_2^2 \right)
  \end{align*}
  The last bound follows because the worst possible lower bound occurs
  at $m=3$.  Hence, there are only a polynomial number of iterations
  of this form as well.
\end{proof}
}{}

\subsection{Exact moments, approximate local optima}\label{sec:partial}
We are now ready to extend the analysis of Theorem \ref{thm:exact} to
the case when we have access to the exact moment tensor, but instead
of using exact moments, we will use \textbf{LocalOpt} with
appropriately chosen $\eps_1$. On the other hand, using a weaker local
optimum algorithm will also give us weaker guarantees on the quality
of the output, giving a weaker form of Lemma \ref{lemma:support}. Over
$\R^n$ (instead of $\sph^{n-1}$), Lemma \ref{lemma:representation}
gives us a formula of the form:
\begin{align*}
  f_m(u) = \norm{u_{V}}^{m}\left(\E{(\T{x_{V}}u_{V}^{0})^{m}}-\gamma_m\right)+
    \norm{u_{W}}^{m}\left(\E{(\T{x_{W}}u_{W}^{0})^{m}}-\gamma_m\right)+\gamma_m \norm{u}^m
\end{align*}
We will optimise the function $g(u) = f_m(u) - \gamma_m \norm{u}^m$
using \textbf{LocalOpt} over the unit sphere. This is equivalent to
optimising $f_m$ and simplifies our derivative calculations.

\begin{lemma}[Exact moments, inexact optima] \label{lemma:exactapprox}
  Let $F=F_V F_W$ have the same first $m-1$ moments as a Gaussian but
  a different $m^{th}$ moment. Suppose we run \textbf{LocalOpt} on
  $g(u)$, starting from a point $u$ where $g(u) \ge \eta$, setting
  $\eps_1 \le m \eta^{2/(m-2)}/M^{2/(m-2)}$. After
  $\poly(n,1/\eps_1,\eta)$ iterations, we will have a point $u^\ast$
  where either $\norm{\pi_V(u^\ast)} \ge 1 - 16 \eps_1$ or
  $\norm{\pi_W(u^\ast)} \ge 1 - 16 \eps_1$.
\end{lemma}
\begin{proof}
  We proceed as in Lemma \ref{lemma:support}. $u^\ast$ lies on a curve
  $C=\{s (u^\ast_V)^0 + t (u^\ast_W)^0: s^2+t^2=1, s\ge 0,t \ge
  0\}$. We will show that neither $s$ nor $t$ is bounded away from 0
  and 1.

  Restricted to the curve $g(u)=g(s,t) = a_v s^m + a_w t^m$. Suppose
  that $a_w \le 0$, then we must have that $s \ge (\eta/M)^{1/m}$. In
  this case, a direct calculation comparing $\innerprod{Dg_u}{u} = m
  a_v s^{m-1} + m a_w t^{m-1}$ with $\norm{Dg_u}=m\sqrt{a_v^2
    s^{2(m-1)} + a_w^2 t^{2(m-1)}}$ will yield $s \ge 1 -
  2\eps_1$. Thus, we may assume that both $a_v$ and $a_w$ are
  positive, and that $a_v \ge a_w$.

  Suppose that for a unit vector $u$, we have $s,t \ge 16 \sqrt{\eps_1}$,
  and the first-order gradient condition:
  \begin{align*}
    \frac{\innerprod{Dg_u}{u}}{\norm{Dg_u}} \ge 1 - \eps_1,
  \end{align*}
  then,
  \begin{align*}
    \lambda_{\max} (D^2g_u) \ge \frac{3m(m-1)g(u)}{4}
  \end{align*}
  (where the eigenvalue is taken only in the subspace orthogonal to
  $u$). Thus, the algorithm continues making progress at such a vector
  $u$. To do this, we lower bound the top eigenvalue by the quadratic
  form in the direction $-tu_V^0 + su_W^0$, which is orthogonal to
  $u$.
  \begin{align*}
    \lambda_{\max} (D^2g_u) & \ge (-tu_V^0 + su_W^0)^T D^2g_u (-tu_V^0
    + su_W^0) \\
    & = m (m-1) ( a_v s^{m-2} t^2 + a_w s^2 t^{m-2} )\\
    & = m(m-1) ( a_v s^{m-1},a_w t^{m-1}) \left( \begin{array}{c}
        t^2/s \\ s^2 /t \end{array} \right)
  \end{align*}
  By construction, $u$ has two nonzero coordinates, taking values $s$
  and $t$ and all other coordinates zero. $Dg_u$ has partial derivatives
  $ma_vs^{m-1}$ and $ma_wt^{m-1}$ in these directions. Thus,
  \begin{align*}
    \frac{\innerprod{Dg_u}{u}}{\norm{Dg_u}} \le \frac{
      \left( \begin{array}{c} a_v s^{m-1} \\ a_w t^{m-1} \end{array}
      \right)^T \left( \begin{array}{c} s \\ t \end{array}
      \right)}{\norm{\left( \begin{array}{c} a_v s^{m-1} \\ a_w
            t^{m-1} \end{array} \right)}}
\end{align*}
Thus we obtain the condition that:
\begin{align*}
    \left( \begin{array}{c} a_v s^{m-1} \\ a_w t^{m-1} \end{array}
    \right) 
    = (1-\eps) \norm{\left( \begin{array}{c} a_v s^{m-1} \\ a_w t^{m-1} \end{array}
    \right)} \left( \begin{array}{c} s \\
        t \end{array} \right) + \sqrt{2\eps - \eps^2} \norm{\left( \begin{array}{c} a_v s^{m-1} \\ a_w t^{m-1} \end{array}
    \right)} r
  \end{align*}
  where $0 \le \eps \le \eps_1$ and $r$ is a unit vector orthogonal to
  $(s,t)$. Substituting this into the previous equation:
  \begin{align*}
    \lambda_{\max}(D^2g_u) & \ge m(m-1) \left[ (1-\eps) \norm{\left( \begin{array}{c} a_v s^{m-1} \\ a_w t^{m-1} \end{array}
    \right)} +
      \sqrt{2\eps-\eps^2} \norm{\left( \begin{array}{c} a_v s^{m-1} \\ a_w t^{m-1} \end{array}
    \right)} r^T \left( \begin{array}{c} t^2 / s \\ s^2 /
        t \end{array} \right) \right] \\
  & \ge m(m-1) \norm{\left( \begin{array}{c} a_v s^{m-1} \\ a_w t^{m-1} \end{array}
    \right)}\left( (1-\eps) - \sqrt{2 \eps-
      \eps^2} \left( \frac{1}{s} + \frac{1}{t}\right)
  \right) \\
  & \ge m(m-1) \norm{\left( \begin{array}{c} a_v s^{m-1} \\ a_w t^{m-1} \end{array}
    \right)} \left( (1-\eps) - 2\sqrt{2\eps-\eps^2} \frac{1}{16
    \sqrt{\eps}}\right) \\
& \ge \frac{3m(m-1)g(u)}{4}
  \end{align*}
  where the last estimate follows from the Cauchy-Schwartz inequality.
\end{proof}

\subsection{Approximate moments and approximate local optima}\label{sec:approx}
By using the robust Schwartz Zippel Lemma (Lemma \ref{lemma:sz})
instead of the usual form, and \textbf{LocalOpt} at Lines 10 and 11 of
\textbf{FindBasis} and Lines 13 and 15 of \textbf{ExtendBasis}, we can
obtain an efficient randomized algorithm. The major difficulty
remaining is that we must bound the error incurred every time we call
\textbf{LocalOpt}. The error analysis is technical: the idea is to
obtain approximate versions of Lemmas \ref{lemma:representation} and
\ref{lemma:support}, and to show that \textbf{LocalOpt} behaves well
on these approximate versions. Consider the first iteration:

\begin{lemma}[Two steps]\label{lemma:onestep}
  Let $x$ have the same first $m-1$ moments as a Gaussian but a
  different $m^{th}$ moment.  Let $u_1 = \sqrt{1-\delta} v_1 -
  \sqrt{\delta} w_1$ be the vector found in the first iteration of
  \textbf{FindBasis}, where $v_1$ and $w_1$ are unit vectors in $V$
  and $W$ respectively.  Suppose we run \textbf{LocalOpt} on
  $g(u)=f_m(u)-\gamma_m \norm{u}^m$ on the orthogonal subspace
  $u_1^\perp$, starting from a point $u$ where $g(u) \ge
  \eta=M\delta^{1/16}$, setting $\eps_1 \le \frac{m
    \eta^{2/(m-2)}}{M^{2/(m-2)}}-60m^2M^2\delta^{5/16}$ as the error
  parameter in \textbf{LocalOpt}. After $\poly(n,1/\eps_1,\eta)$
  iterations, we will have a point $u^\ast$ where either
  $\norm{\pi_V(u^\ast)} \ge 1 - \delta^{1/8}$ or $\norm{\pi_W(u^\ast)}
  \ge 1 - \delta^{1/8}$
\end{lemma}
\ifthenelse{\boolean{longversion}}{ 
Tthe sequence of ideas in this proof is not unlike the proofs in
Section \ref{sec:basis}: first we derive a nice representation of
$f_m$ (cf Lemma \ref{lemma:representation}, then we analyse the
support of a local optimum under this representation (cf Lemma
\ref{lemma:support}) -- we are not able to claim that the local
optimum found is contained wholly in $V$ or $W$, but since we are
satisfied with approximate local optima, we can bound the components
around 0 and 1. All through this, we must bound the error, and try to
push through the calculations of Lemma \ref{lemma:exactapprox}.

\begin{proof}[Proof of Lemma \ref{lemma:onestep}]
  First, we will construct an orthonormal basis which includes $u_1$:
  extend $\{v_1\}$ and $\{w_1\}$ to orthonormal bases $\{v_i\}$ and
  $\{w_i\}$ of $V$ and $W$ respectively. Replace $v_1$ and $w_1$ with
  the following two vectors:
\begin{align*}
u_1 & = \sqrt{1-\delta} v_1 - \sqrt{\delta} w_1 \\
\hat{u_1} & = \sqrt{\delta} v_1 + \sqrt{1-\delta} w_1 \\
\end{align*}
Thus our basis will be $\{u_1,\hat{u}_1,v_2,\ldots,v_k,
w_2,\ldots,w_l\}$. For a vector $x=(x_1,\ldots,x_n)$ in the basis of
$\{v_i\}$ and $\{w_i\}$, we now have:
\begin{align*}
x=( \sqrt{1-\delta}x_1 - \sqrt{\delta} x_2, \sqrt{\delta} x_1 +
\sqrt{1-\delta} x_2,x_3,\ldots,x_n)
\end{align*}
which is simply a rotation (unitary transformation) in the first two
coordinates.

We evaluate the $m^{th}$ moment on the subspace orthogonal to
$u_1$. Let $\xi$ be a point on this orthogonal subspace: note that
$\xi$ has 0 component in the first coordinate:
\begin{align*}
f_m(\xi) & = \E{(\T{x}\xi)^m} \\
 &= \E{ \left( \sqrt{\delta} x_1 \xi_2 + \sqrt{1-\delta} x_2 \xi_2 + \sum_{i=2}^{k}
   x_{v_i} \xi_{v_i} + \sum_{i=2}^{l} x_{w_i} \xi_{w_i} \right)^m} \\
\end{align*}
We can break up the argument into two dot products, which are
independent of each other. Moreover, observe that the norm of the two
constituent parts of the $\xi$ vector taken together is still 1.
\begin{align*}
  \T{x}\xi = \T{(x_1,x_{v_2},\ldots,x_{v_k})}(\sqrt{\delta} \xi_2,
  \xi_{v_2},\ldots,\xi_{v_k}) +
  \T{(x_2,x_{w_2},\ldots,x_{w_l})}(\sqrt{1-\delta}\xi_2,\xi_{w_2},\ldots,\xi_{w_l})
\end{align*}
Hence, we can apply Lemma \ref{lemma:representation}: this gives a
perturbed version of Lemma \ref{lemma:representation}.
\begin{align*}
  f_m(\xi) = \left( \delta \xi_2^2 + \sum_{i=2}^k \xi_{v_i}^2
  \right)^{m/2} \E{(\T{(x_1,x_{v_2},\ldots,x_{v_k})}(\sqrt{\delta}
    \xi_2, \xi_{v_2},\ldots,\xi_{v_k})^0)^m-\gamma_m} + \\
  \left( (1-\delta) \xi_2^2 + \sum_{i=2}^l \xi_{w_i}^2 \right)^{m/2}
  \E{\left( \T{(x_2,x_{w_2},\ldots,x_{w_l})}
      (\sqrt{1-\delta}\xi_2,\xi_{w_2},\ldots,\xi_{w_l})^0 \right)^m -
    \gamma_m} + \gamma_m
\end{align*}

Fixing a point $\xi^\ast \in u_1^\perp \cap \mathbb{S}^{n-1}$: we will
give a curve $C$ which passes through this point and remains on the
unit sphere. We will analyse the value of the $g(\xi) = f_m(\xi) -
\gamma_m \norm{\xi}^m$ on this curve -- as before, every point which is
a local optimum on $\mathbb{S}^{n-1}$ has to be a local optimum on $C$
as well. Thus by studying the local optima over $C$, we will be able
to describe the strucure of the local optima in full space.

We may assume that all the $\xi^\ast_i$ are nonnegative -- otherwise
we can pick simply negate the associated basis vector. We take the
following as the components for $C$:
\begin{align*}
  \xi_v^\ast &= \frac{1}{\sqrt{\sum_{i=2}^k (\xi_{v_i}^\ast)^2}}(
  0,0,\xi_{v_2}^\ast,\ldots,\xi_{v_k}^\ast,0,\ldots,0) \\
  \xi_w^\ast & = \frac{1}{\sqrt{(1-\delta)(\xi_2^\ast)^2+\sum_{i=2}^l (\xi_{w_i}^\ast)^2}} (
  0,\sqrt{1-\delta}\xi_2^\ast,0,\ldots,0,\xi_{w_2}^\ast,\ldots,\xi_{v_l}^\ast)\\
  \xi_1^\ast & = ( 1, 0, \ldots, 0 )
\end{align*}
Since these are the only three directions that change along $C$, we
will use these three vectors as an orthonormal basis. Now, defining
the following quantity:
\begin{align*}
  \alpha = (\xi_2^\ast)^2 / \left((1-\delta)(\xi_2^\ast)^2 +
    \sum_{i=2}^l (\xi_{w_i}^\ast)^2 \right)
\end{align*}
we can write our curve $C$ as:
\begin{align*}
  C = \{ y \xi_v^\ast + z \xi_w^\ast + \sqrt{\alpha \delta} z
  \xi_1^\ast:y^2+(1+\delta \alpha )z^2=1,y,z\ge 0\}
\end{align*}
Specifically, we will use the basis $\xi_v^\ast$ and $(1+\alpha
\delta)^{-1} (\xi_w^\ast + \xi_1^\ast)$. Note that in this basis, $y$
is precisely the coordinate along the first basis vector and
$(1+\delta \alpha)^{1/2}z$ is the coordinate along the second basis
vector. Denote this latter quantity by $z'$, then by the chain rule,
we have that $\partial / \partial z' = (1+\delta
\alpha)^{-1/2} \partial / \partial z$. 

Restricted to $C$, the expectation terms simplify: note that
$(\sqrt{1-\delta}\xi_2,\xi_{w_2},\ldots,\xi_{w_l})^0$ remains constant
on $C$, so the second expectation term reduces to a constant, which we
will denote with $a_w$. The first expectation term does not remain
constant, because there is an additional component in the direction of
$v_1$, but this component always has a small magnitude. With a change
of basis, we can simplify this expression to involving only $y$ and
$z$:
\begin{align*}
  \E{(\T{(x_1,x_{v_2},\ldots,x_{v_k})}(\sqrt{\delta}
    \xi_2, \xi_{v_2},\ldots,\xi_{v_k})^0)^m} = \E{\left(
      (x_1,x_{\xi_v^\ast})^T(\sqrt{\alpha \delta} z, y) \right)^m}
\end{align*}
We will denote the first expectation term by $a_v$. In full, our
objective function on $C$ is given by:
\begin{align*}
  g(\xi) & = [ \delta \alpha z^2 + y^2 ]^{(m/2)} \E{\left(
      \T{(x_1,x_v)}(\sqrt{\alpha \delta}z,y)^0 \right)^m -
    \gamma_m } + a_w z^m \\
  & = [ \delta \alpha z^2 + y^2 ]^{(m/2)} a_v(y, z) +
  a_w z^m
\end{align*}
Next we will examine the local optima on $C$: let $\xi$ be the output
of of \textbf{LocalOpt}: we will show that $\xi$ has large projection
with either the $V$ or $W$ subspace (cf Lemma \ref{lemma:support}). We
will analyse the following cases:
\begin{enumerate}
\item $y^2 \le \delta^{1/4}$ or $z^2 \le \delta^{1/4}$.
\item $y^2 \ge \delta^{1/4}$ and $z^2 \ge 1/3$.
\item $z^2 \ge \delta^{1/4}$ and $y^2 \ge 1/3$.
\end{enumerate}
Case 1: Suppose that $y^2 \le \delta^{1/4}$, then we must have
$z \ge \sqrt{1-\delta^{1/4}-\alpha\delta}$. The approximate local optimum $u$
that we compute has projection at least $\sqrt{1-\delta}$ on this
local optimum, and hence, the projection of $u$ onto $w$ is at least:
\begin{align*}
  \norm{\pi_W(u)} & \ge \sqrt{(1-\delta)(1-\delta^{1/4}-\alpha\delta)} -
  \sqrt{\delta} \\
  & \ge 1-\delta/2 - \delta^{1/4}/2 - \alpha\delta - \sqrt{\delta} \\
  & \ge 1 - \delta^{1/4}
\end{align*}
In this case, for sufficiently small $\delta$, we have:
\begin{align*}
  \norm{\pi_W(u)}^2 \ge 1- \delta^{1/8}
\end{align*}
The argument for when $z^2 \le \delta^{1/4}$ is identical.  

Case 2: We will prove that \textbf{LocalOpt} can not terminate in this
region by carrying out the calculations of Lemma
\ref{lemma:exactapprox} whilst keeping track of errors. Thus, let
$\xi$ be a point in this range, we will show that if the first
derivative condition in \textbf{LocalOpt} is satisfied, then the
second derivative condition is unsatisfied, thus \textbf{LocalOpt} can
not terminate at $\xi$. First, let us examine how $f$ changes over
$C$:
\begin{claim}[First partials under perturbations] \label{lemma:one}
  In the range where $y^2 \ge \delta^{1/4}$ and $z^2 \ge 1/3$, 
  \begin{align*}
    \abs{\frac{\partial g}{\partial y} - ma_vy^{m-1}} & \le 3 m M
    \sqrt{\delta} \\
    \abs{ \frac{\partial g}{\partial z} - m a_w z^{m-1}}& \le 4 m M
    \sqrt{\delta}
  \end{align*}
\end{claim}
As a corollary, via the triangle inequality, we have that:
\begin{align*}
  \norm{ (g_y,g_{z})} \ge m \norm{ (a_v y^{m-1}, a_w z^{m-1})} - 5 mM
  \sqrt{\delta}
\end{align*}

\begin{claim}[Second partials under perturbations] \label{lemma:two}
  In the range where $y^2 \ge \delta^{1/4}$ and $z^2 \ge 1/3$:
  \begin{align*}
    \abs{ \frac{\partial^2 g}{\partial y^2} - m(m-1)a_v
      y^{m-2}} & \le
    c_{vv} m^2 M \sqrt{\delta} \\
    \abs{ \frac{\partial^2 g}{\partial z^2} - m(m-1)a_w z^{m-2}} &
    \le
    c_{ww} m^2 M \sqrt{\delta} \\
    \abs{ \frac{\partial^2 g}{\partial y \partial z}} & \le c_{vw}
    m^2 M \sqrt{\delta}
  \end{align*}
  where $c_{vv}$, $c_{vw}$ and $c_{ww}$ are absolute constants bounded
  by 20.
\end{claim}
Throughout the rest of this calculation, we will use the basis of
$n-1$ vectors consisting of $\{\xi_v^\ast$, $(1+\alpha \delta)^{-1}
(\xi_w^* + \xi_1^\ast)\}$, and any orthonormal extension to
$u_1^\perp$. In particular, in this basis, $\xi =
(y,z',0,\ldots,0)$.

As before, we will lower bound the contribution of the second
derivative term. Our direction of movement will be
$(-z',y,0,\ldots,0)$. This vector is clearly a unit vector orthogonal
to $\xi$.
where top eigenvalue is taken orthogonal to $\xi$.
\begin{align*}
  \lambda_{\max} (D^2g_\xi) & \ge (-z',y) D^2g_\xi
  \left( \begin{array}{c} -z' \\ y \end{array} \right) \\
&  \ge
  ( - \sqrt{1 + \alpha \delta} z, y)
\left( 
\begin{array}{cc}
g_{yy} & g_{z'y} \\
g_{z'y} & g_{z'z'}
\end{array}
\right)
\left(
\begin{array}{c}
- \sqrt{1+\delta \alpha}z \\
y
\end{array}
\right)
\end{align*}
We can further use Claim \ref{lemma:two} to simplify the other
components of the quadratic form:
\begin{align*}
  \lambda_{\max} (D^2g_\xi) & \ge  (1+\delta \alpha) z^2 g_{yy} +
  y^2 g_{z'z'} - 2c_{vw}m^2M\sqrt{\delta}\\
& \ge (1+\delta \alpha )m(m-1) ( a_v y ^{m-1}, a_w z^{m-1} ) \left( \begin{array}{c} z^2
    / y \\ y^2 / z \end{array} \right) - (1+\delta \alpha)(c_{zz} + c_{yy} + 2
c_{zw})m^2 M \sqrt{\delta}
\end{align*}
Our first derivative condition is given by:
\begin{align*}
  \frac{\innerprod{Dg_\xi}{\xi}}{\norm{Dg_\xi}} \ge 1 - \eps_1
\end{align*}
Since $\xi=(y,z',0,\ldots,0)$ has only two nonzero components, we need
only evaluate two components of the derivative: furthermore, we can
lower bound the norm $\norm{Dg_\xi} \ge \norm{ (g_y, g_{z'}) }$, which
gives the following lower bound:
\begin{align*}
  \frac{(g_y,g_{z'}) \left( \begin{array}{c} y \\ z' \end{array} \right)}{\norm{( g_y, g_{z'})}} \ge 1 - \eps_1
\end{align*}
Rearranging, and applying Claim \ref{lemma:one} yields:
\begin{align*}
  m (a_v y^{m-1}, a_w z^ {m-1}) \left( \begin{array}{c} y \\
      z \end{array} \right) & \ge (1 - \eps_1) \norm{(g_y, g_{z'})} - 7
  mM \sqrt{\delta} \\
  & \ge m (1-\eps_1) \norm{ (a_v y^{m-1}, a_w z^{m-1})} - 12 mM
  \sqrt{\delta} \\
  & \ge m ( 1 -\eps_1 - \frac{12M\sqrt{\delta}}{\eta} ) \norm{ (a_v y^{m-1}, a_w z^{m-1})}
\end{align*}
Thus, we can rewrite this relationship for unit vector $r$ orthogonal
to $(a_v y^{m-1}, a_w z^ {m-1})$ and $0 \le \eps \le \eps_1 +
\frac{12M\sqrt{\delta}}{\eta}$:
\begin{align*}
\left( \begin{array}{c} a_v y^{m-1} \\ a_w z^{m-1} \end{array} \right)
= (1-\eps) \norm{ (a_v y^{m-1}, a_w z^{m-1})} \left( \begin{array}{c} y \\ z \end{array} \right) + \sqrt{2\eps-\eps^2} \norm{ (a_v y^{m-1}, a_w z^{m-1})}r
\end{align*}
Substituting this back into our lower bound for $\lambda_{\max}$
yields:
\begin{align*}
\lambda_{\max} & \ge (1+\delta \alpha) (1- \eps)\norm{ (a_v y^{m-1}, a_w
  z^{m-1})} (z^2 + y^2) - \sqrt{2\eps-\eps^2} \norm{ (a_v y^{m-1}, a_w
  z^{m-1})} \left( \frac{1}{y} + \frac{1}{z} \right) \\
& \qquad -80 m^2M \sqrt{\delta} \\
& \ge (1+\delta \alpha ) ( 1 - \delta^{1/6}) m (m-1)f(\xi) - \sqrt{2}
\delta^{1/24} - 80 m^2M \sqrt{\delta} \\
& \ge \frac{3}{4} m( m-1) f(\xi)
\end{align*}
where we used the Cauchy-Schwartz inequality for:
\begin{align*}
  \norm{ ( a_v y^{m-1}, a_w z_m-1)} \ge a_v y^m + a_w z^m \ge g(\xi) -
  m M \sqrt{\alpha \delta} 
\end{align*}
Case 3: This case follows from the exactly the same analysis as
above. It is in fact substantially easier, as the denominator terms
$\alpha \delta z + y$ are in fact all bounded by constants now, and
hence the numerator is small enough in almost all cases above to bound
the terms.

We now provide the proofs for the claims regarding the coefficients
$a_v$ and $a_w$. In explicitly taking derivatives, it is important to note the
following:
\begin{align*}
  a_v &= \E{\left( \T{(x_1,x_v)}(\sqrt{\alpha \delta}z,y)^0 \right)^m}
  -
  \gamma_m  \\
  &= \frac{1}{(\alpha \delta z^2 + y^2)^{(m/2)}} \E{\left(
      \T{(x_1,x_v)}(\sqrt{\alpha \delta}z,y) \right)^m} - \gamma_m
\end{align*}
For ease of notation, denote $\phi = (\sqrt{\alpha \delta}z, y)$, we
will suppress all but one $\phi$ argument in our moment tensors, thus
we will write $A(\phi)$ instead of $A(\phi,\ldots,\phi)$, and $A(\phi,
e_1)$ instead of $A(\phi,\ldots,\phi,e_1)$. If $A$ is a $m^{th}$ order
tensor, its derivative has components given by $ (D\hat{A}_\phi
)_i=mA(\phi,e_i) $ where $A$ takes $(m-1)$ copies of $\phi$. We also
have the Hessian $D^2$: $ (D^2A_\phi)_{ij} = m ( m- 1)
A(\phi,e_i,e_j)$. We can bound the spectral norm of $D^2A$ using
Claim \ref{claim:spectral}, which yields $\lambda_{max} (D^2A)
\le m(m-1)M$.
\begin{proof}[Proof of Claim \ref{lemma:one}]
Firstly, we have:
\begin{align*}
  \frac{\partial g}{\partial y} & = my(\delta \alpha z^2 +
  y^2)^{(m/2)-1}a_v + (\delta \alpha z^2 + y^2)^{(m/2)}
  \frac{\partial a_v}{\partial y} \\
  \frac{\partial g}{\partial z} & =mz^{m-1}a_w + m\alpha \delta z(\delta
  \alpha z^2 + y^2)^{(m/2)-1} a_v + (\delta \alpha z^2 + y^2)^{(m/2)}
  \frac{\partial a_v}{\partial z}
\end{align*}
The $m \alpha \delta z(\delta \alpha z^2 + y^2)^{(m/2)-1} a_v$ is
upper bounded in absolute value in $mM \delta$. Similarly, it is also
clear that:
\begin{align*}
  \abs{my(\delta \alpha z^2 + y^2)^{(m/2)-1}a_v - ma_v y^{m-1}} \le
  mM\sqrt{\delta}
\end{align*}
Thus it remains to show that the partial derivative terms are
small:
\begin{align*}
  (\delta \alpha z^2 + y^2)^{(m/2)}\frac{\partial a_v}{\partial y} &
  =(\delta \alpha z^2 + y^2)^{(m/2)}\left[ \frac{-my}{(\delta \alpha
      z^2 + y^2)^{(m/2)+1}}A(\phi,\phi) + \frac{m}{(\delta \alpha z^2
      + y^2)^{(m/2)}}A(\phi,e_1)\right]\\
  & = m \left(
    \frac{-y\sqrt{\alpha \delta} zA(\phi,e_2)
      -y^2A(\phi,e_1)}{\alpha \delta z^2 + y^2}+A(\phi,e_1)\right)\\
  & = m \frac{-y
    \sqrt{\alpha \delta}zA(\phi,e_2) + \alpha \delta z^2
    A(\phi,e_1)}{\alpha
    \delta z^2 + y^2}
\end{align*}
When we have a term like $A(\phi,\ldots,\phi,e_1)$, the arguments are
not normalised. In particular:
\begin{align*}
  A(\phi,\ldots,\phi,e_1) = (\delta \alpha z^2 +
    y^2)^{(m-1)/2} A(\phi^0,\ldots,\phi^0,e_1)
\end{align*}
Thus, normalising gives:
\begin{align*}
  \abs{(\delta \alpha z^2 + y^2)^{(m/2)}\frac{\partial a_v}{\partial
      y}} & \le mM\sqrt{\delta} + m\delta M \\
  & \le 2mM\sqrt{\delta}
\end{align*}
For the other partial derivative, we want to compute:
\begin{align*}
  (\delta \alpha z^2 + y^2)^{(m/2)} \frac{\partial a_v}{\partial z} &
  = (\delta \alpha z^2 + y^2)^{(m/2)}\left[ \frac{-m\alpha \delta
      z}{(\alpha \delta z^2 + y^2 )^{(m/2)+1}} A(\phi,\phi) + \frac{m
      \sqrt{\alpha
        \delta}}{(\alpha \delta z^2 + y^2 )^{m/2}}A(\phi,e_2) \right]\\
  & = m\sqrt{\alpha \delta} \left( \frac{-\sqrt{\alpha \delta} z
      A(\phi,\phi)+\alpha \delta z^2A(\phi,e_2)+y^2A(\phi,e_2)}{\alpha
      \delta z^2 + y^2}\right)
\end{align*}
Applying the same method:
\begin{align*}
  \abs{(\delta \alpha z^2 + y^2)^{(m/2)} \frac{\partial a_v}{\partial
      z}} \le 3mM\sqrt{\delta}
\end{align*}
Hence combining this with our earlier bound, we have the desired
inequality.
\end{proof}

\begin{proof}[Proof of Claim \ref{lemma:two}]
  By direct calculation, we obtain:
\begin{align*}
  \frac{\partial^2 g}{\partial y^2} & = (\delta \alpha z^2 +
  y^2)^{(m/2)}\frac{\partial^2 a_v}{\partial y^2} + 2my (\delta \alpha
  z^2 + y^2)^{(m/2)-1} \frac{\partial a_v}{\partial y} + ma_v (\delta
  \alpha z^2 + y^2)^{(m/2)-2} (\delta \alpha z^2 + (m-1)y^2)
\end{align*}
We now estimate the three terms in this sum -- the first two terms
will be of order $\sqrt{\delta}$, and the last term will give us
approximately $m(m-1)a_vy^{m-2}$.
\begin{align*}
  & (\delta \alpha z^2 + y^2)^{(m/2)}\frac{\partial^2
    a_v}{\partial
    y^2} \\
  &= (\delta \alpha z^2 + y^2)^{(m/2)} \left\{
    \frac{-m^2y}{(\alpha \delta z^2 + y^2)^{(m/2)+1}}
    \left[ \frac{-y\sqrt{\alpha \delta} zA(\phi,e_2)}{\alpha
        \delta z^2 + y^2} + \frac{\alpha \delta
        z^2A(\phi,e_1)}{\alpha \delta z^2 + y^2} \right] +
    \frac{m}{(\alpha \delta z^2 + y^2)^{m/2}} \left[
      \frac{-\sqrt{\alpha
          \delta } z A(\phi,e_2)}{\alpha \delta z^2 + y^2}\right. \right.\\
  & \left. \left.+ \frac{-(m-1)y\sqrt{\alpha \delta} z
        A(e_2,e_1)}{\alpha \delta z^2 + y^2} +
      \frac{y^2\sqrt{\alpha \delta}zA(\phi,e_2)}{(\alpha
        \delta z^2 + y^2)^2}+\frac{\alpha \delta z^2 (m-1)
        A(e_1,e_1)}{\alpha \delta z^2 +
        y^2}+\frac{-2y\alpha \delta z^2 A(\phi,
        e_1)}{(\alpha \delta z^2
        + y^2)^2}\right]\right\} \\
  & = (-m^2y)\left[ \frac{-y\sqrt{\alpha \delta}
      zA(\phi,e_2)}{(\alpha \delta z^2 + y^2)^2} +
    \frac{\alpha \delta z^2A(\phi,e_1)}{(\alpha \delta z^2 +
      y^2)^2} \right] + m\left[ \frac{-\sqrt{\alpha \delta } z
      A(\phi,e_2)}{\alpha \delta z^2 + y^2} +
    \frac{-(m-1)y\sqrt{\alpha \delta} z A(e_2,e_1)}{\alpha
      \delta z^2
      + y^2}\right. \\
  & + \left.  \frac{y^2\sqrt{\alpha
        \delta}zA(\phi,e_2)}{(\alpha \delta z^2 +
      y^2)^2}+\frac{\alpha \delta z^2 (m-1) A(e_1,e_1)}{\alpha
      \delta z^2 + y^2} + \frac{-2y\alpha \delta z^2
      A(\phi, e_1)}{(\alpha \delta z^2 + y^2)^2}\right]
\end{align*}
We will bound the magnitude of every term in this sum. Consider the
first term of the form:
\begin{align*}
  \abs{(-m^2y) \frac{-y\sqrt{\alpha \delta}
      zA(\phi,e_2)}{(\alpha \delta z^2 + y^2)^2}} \le m^2
  \abs{\frac{y^2 \sqrt{\delta}A(\phi,e_2)}{(\alpha \delta z^2+y^2)^2}}
\end{align*}
Thus, since $m \ge 3$:
\begin{align*}
  \abs{(-m^2y) \frac{-y\sqrt{\alpha \delta}
      zA(\phi,e_2)}{(\alpha \delta z^2 + y^2)^2}} \le 3m^2 M
  \abs{\frac{y^2 \sqrt{\delta}}{\alpha \delta z^2+y^2}}
\end{align*}
Now, $y^2 / (\alpha \delta z^2 + y^2) \le 1$, hence:
\begin{align*}
    \abs{(-m^2y) \frac{-y\sqrt{\alpha \delta}
      zA(\phi,e_2)}{(\alpha \delta z^2 + y^2)^2}} \le  3 m^2 M\sqrt{\delta}
\end{align*}
Of the seven terms in the sum, the first, third and fifth terms can be
analysed exactly as above, and their sum can be upper bounded by
$15m^2M\sqrt{\delta}$. For the remaining terms we have to use our
lower bound on $y$, for example:
\begin{align*}
  \abs{(-my)\frac{\alpha \delta z^2A(\phi,e_1)}{(\alpha \delta z^2 +
      y^2)^2}} & \le mM\abs{\frac{\delta y}{\alpha \delta z^2 +
      y^2}} \\
 & \le mM\abs{\frac{\delta}{y}} \\
 & \le mM \delta^{7/8}
\end{align*}
By this reasoning, we can bound all seven terms by
$m^2M\sqrt{\delta}$, hence this term in $\partial^2 g/\partial y^2$
contributes is bounded in absolute value by $7m^2M\sqrt{\delta}$. For
the second term in that expression, the analysis is almost identical
to the previous claim and gives 
\begin{align*}
  \abs{2my (\delta \alpha z^2 + y^2)^{(m/2)-1} \frac{\partial
      a_v}{\partial y}} & = 2m^2 \abs{y (\delta \alpha z^2 +
    y^2)^{(m/2)-1} \frac{(-y \sqrt{\alpha \delta}zA(\phi^0,e_2) +
      \alpha \delta z^2
      A(\phi^0,e_1))}{(\delta \alpha z^2 + y^2)^{3/2}}} \\
  & \le 2m^2\abs{ y\frac{(-y \sqrt{\alpha \delta}zA(\phi^0,e_2) +
      \alpha \delta z^2 A(\phi^0,e_1))}{(\delta \alpha z^2 + y^2)}} \\
& \le 2m^2M\sqrt{\delta} + 2m^2M \abs{\frac{\delta}{y}} \\
& \le 4m^2M \sqrt{\delta}
\end{align*}
Thus, we have:
\begin{align*}
  \abs{\frac{\partial^2 g}{\partial y^2} - ma_v (\delta \alpha z^2
    + y^2)^{(m/2)-2} (\delta \alpha z^2 + (m-1)y^2)} \le
  19m^2M\sqrt{\delta}
\end{align*}
By applying the triangle inequality:
\begin{align*}
   \abs{  ma_v (\delta \alpha z^2 + y^2)^{(m/2)-2} (\delta \alpha z^2
    + (m-1)y^2) - m(m-1)a_v y^{m-2}} \le m^2M\sqrt{\delta}
\end{align*}
Thus we have the desired result for the second partial with respect to
$y$. The other second derivatives are computed in a similar way.
\end{proof}
\end{proof}
Using the above, we are now examine what happens after $t$ iterations
of \textbf{FindBasis}.  }{} The following theorem shows that after $k$
iterations of \textbf{FindBasis}, our error blows up at most doubly
exponentially in $k$. The proof holds for \textbf{ExtendBasis} is as
well.
\begin{theorem}[Multiple iterations]\label{thm:approximate}
  Suppose \textbf{FindBasis} finds $j \le k$ orthogonal vectors
  $\{u_1,\ldots,u_j\}$ of $g(u)$ taking $\eps_1$ such that $\eta \le
  M\eps_1^{1/16^j}$ for each call of \textbf{LocalOpt}, then
  $\norm{\pi_V(u_j)}^2 \ge 1-\epsilon_1^{(1/16)^j}$.
\end{theorem}

\ifthenelse{\boolean{longversion}}{
\begin{proof}[Proof of Theorem \ref{thm:approximate}]
After $t$ iterations, we have a basis of orthonormal vectors
$\{u_1,\ldots,u_t\}$ where each $u_i$ is close to some vector in
$V$:
\begin{align*}
u_1 & = a_{11} v_1 + b_{11} w_1 \\
u_2 & = a_{21} v_1 + a_{22} v_2 + b_{21} w_1 + b_{22} w_2 \\
\vdots & \qquad \vdots \\
u_t &= a_{t1} v_1 + \cdots + a_{tt} v_t + b_{t1} w_1 + \cdots b_{tt} w_t
\end{align*}
We use the orthonormal basis $\{u_i\}$, $\{v_{t+1},\dots,v_k\}$, the
remaining vectors in $W$ $\{w_{t+1},\ldots,w_{n-k}\}$, and approximate
copies of $\{w_1,\ldots,w_t\}$. This last set is given by:
\begin{align*}
w_1' &= c_1 w_1 + \sum_{i=1}^t d_{1i} v_i + \sum_{i=1}^t e_{1i} w_i \\
\vdots & \qquad \vdots \\
w_t' &= c_t w_t + \sum_{i=1}^t d_{ti} v_i + \sum_{i=1}^t e_{ti} w_i \\
\end{align*}
In these sums we have $d_{ii} = e_{ii} = 0$, and we have
orthonormality between these vectors. Consider the inner product $x^T
\xi$, where $\xi$ is of unit length and lies in the space orthogonal
to $\{u_i\}$:
\begin{align*}
  x^T \xi & = \sum_{i=t+1}^k x_{v_i} \xi_{v_i} + \sum_{i=t+1}^{n-k}
  x_{w_i} \xi_{w_i} + \sum_{i=1}^{t} \xi_{w_i'} ( c_i x_{w_i} +
  \sum_{j=1}^i d_{ij} x_{v_j} + \sum_{j=1}^i e_{ij} x_{w_j} ) \\
  & = (x_{v_1},\ldots,x_{v_t},x_{v_{t+1}},\ldots,x_{v_k})^T
  (\sum_{i=1}^t \xi_{w_i'} d_{i1}, \ldots, \sum_{i=1}^t \xi_{w_i'}
  d_{it}, \xi_{v_{t+1}},\ldots, \xi_{v_{k}})+ \\
  & \qquad + (x_{w_1},\ldots,x_{w_t},x_{w_{t+1}},\ldots,x_{w_{n-k}})^T
  (\xi_{w_1'} c_1 + \sum_{i=1}^t \xi_{w_i}' e_{i1}, \xi_{w_t'} c_t +
  \sum_{i=1}^t \xi_{w_i'} e_{it}, \xi_{w_{t+1}},\ldots, \xi_{w_{n-k}})
\end{align*}
The two vectors formed from $\xi$ have total norm 1.
Now, we can apply Lemma \ref{lemma:representation}, to obtain:
\begin{align*}
  f_m(\xi') = \left( \sum_{j=t+1}^k \xi_{v_j}^2 + \sum_{j=1}^t
    \left(\sum_{i=1}^t \xi_{w_i'} d_{ij}\right)^2 \right)^{m/2} a_v +
  \left( \sum_{j=t+1}^{n-k} \xi_{w_j}^2 + \sum_{j=1}^t
    \left(\xi_{w_j'} c_j + \sum_{i=1}^t \xi_{w_i'} e_{ij}\right)^2
  \right)^{m/2} a_w + \gamma_m
\end{align*}
where the expectation term $a_v$ is given by:
\begin{align*}
a_v = \E{ \left[ (x_{v_1},\ldots,x_{v_t},x_{v_{t+1}},\ldots,x_{v_k})^T
  (\sum_{i=1}^t \xi_{w_i'} d_{i1}, \ldots, \sum_{i=1}^t \xi_{w_i'}
  d_{it}, \xi_{v_{t+1}},\ldots, \xi_{v_{k}})^0\right]^m - \gamma_m}
\end{align*}
(and similarly for $a_w$). As in the single iteration case, we
restrict to a curve. Fix an output $\xi^\ast$ of \textbf{FindBasis}:
we will fix the ratio of the components $\{\xi_{w_j}\}$ in the ratio
of $\xi^\ast$, and similarly, we will fix the ratios of
$\{\xi_{w_1'},\ldots,\xi_{w_t'},\xi_{w_{t+1}},\ldots,\xi_{w_{n-k}} \}$
according to $\xi^\ast$ as well. This gives the following restriction
on our curve after subtracting $\gamma_m \norm{\xi}^m$.
\begin{align*}
g(\xi') = a_v \left[ (y')^2 + (z')^2\left( \sum_{j=1}^t
     \left(\sum_{i=1}^t d_{ij} \xi_{w_i'}^\ast/l \right)^2 \right) \right]^{m/2} + a_w (z')^m
\end{align*}
where $l$ is a constant given by:
\begin{align*}
  l = \frac{1}{\left( \sum_{j=t+1}^{n-k} (\xi_{w_j}^\ast)^2 +
      \sum_{j=1}^t \left(\xi_{w_j'}^\ast c_j + \sum_{i=1}^t
        \xi_{w_i'}^\ast e_{ij}\right)^2 \right)}
\end{align*}
The coefficient of $z'^2$ is bounded by at most
$2t(\eps_1^{1/16})^t$, hence using the previous lemma for a single
iteration, the output produced here is a $(t+1)^{th}$ vector $u_{t+1}$
such that:
\begin{align*}
  \innerprod{u_{t+1}}{u^\ast} &  
  \ge 1 - \left( 2t ( \eps_1^{1/16})^t \right)^{1/8} \\
  & \ge 1 - (\eps_1^{1/16})^{t+1}
\end{align*}
for sufficiently small $\eps_1$ (relative to $k$).
\end{proof}
}{}\subsection{Algorithms}
Using \textbf{LocalOpt}, we have an algorithm for factoring (Problem
\ref{problem:factoring}). To deal with the errors introduced by
sampling and approximate local optima, we replace the Schwartz-Zippel
step in \textbf{FindBasis} with the robust version in Lemma
\ref{lemma:sz}, where we set the error parameter of the robust
Schwartz-Zippel lemma to be $(\eta-\norm{M^m}_2\eps)/n^m$.  We will
use the following robust version of the Schwartz-Zippel identity test.
\begin{lemma}[Robust Schwartz-Zippel]\label{lemma:sz}
  Let $p$ be a degree $m$ polynomial over $n$ variables and $K$ a
  convex body in $\R^n$. If there exists $x \in K$ such that
  $\abs{p(x)} > \eps (2cn)^m$, then for $l$ random points $s_i$,
  $\prob{\forall s_i: \abs{p(s_i)}\le \eps} \le 2^{-l}$.
\end{lemma}
\ifthenelse{\boolean{longversion}}{
\subsection{Robust Schwartz-Zippel lemma}\label{sec:schwartz}

\begin{proof}[Proof of Lemma \ref{lemma:sz}]
  Let $\mu$ denote the uniform measure over $K$, by Corollary 2 of
  Carbery and Wright \cite{Carbery}:
\begin{align*}
  \max_{ x \in K} \abs{p(x)}^{1/m} \eps^{-1/m} \mu(\{x\in
  K: \abs{p(x)} \le \eps\}) \le cn
\end{align*}
Consider our $l$ samples -- there are two possibilities:
\begin{enumerate}
\item $\mu(\{x\in K: \abs{p(x)}\le \eps\}) \ge 1/2$. In this case, we
  have $\abs{p(x)} \le \eps(2cn)$ from the bound above.
\item $\mu(\{x\in K: \abs{p(x)}\le \eps\}) \le 1/2$. Then,
  $\prob{\forall i\abs{p(x_i)}\le \eps} \le 1/2^l$.
\end{enumerate}
\end{proof}
We can of course amplify this probability by repeating the test (or
simply taking $l$ larger).  }{ Our running time includes the number of
samples needed to estimate the moment tensor to within error
$\eps$. This sample complexity problem is studied in the random
matrices literature for special distributions. The culmination of one
line of results \cite{adamczak, sriver, vershynin} shows that for
logconcave distributions $N = C_F (n, \eps) \le C_\eps n$ samples
suffice to achieve
\begin{align*}
\E{\norm{\frac{1}{N}\sum_{i=1}^N x_i x_i^T - I}} \le \eps.
\end{align*}
For higher moments, Guedon and Rudelson \cite{guedon} proved that that
$O(n^{m/2} \log(n))$ samples are sufficient to approximate moments in
all directions up to a $1 + \eps$ factor.
}

\ifthenelse{\boolean{pictures}}{
\begin{algorithm}
  \caption{FactorUnderGaussian}
  \begin{algorithmic}[1]
    \REQUIRE Highest moment $m$, distribution $F$.
    \STATE $B \leftarrow \boldsymbol{FindBasis(m,F)}$.
    \STATE $U \leftarrow \boldsymbol{ExtendBasis(m,F,B,\phi)}$.
    \RETURN $U$
  \end{algorithmic}
\end{algorithm}
}{}

\ifthenelse{\boolean{longversion}}{
\begin{proof}[Proof of Theorem \ref{thm:factoring-gaussian}]
  We choose $\eps_1$ (the first step local iteration) to be:
\begin{align*}
  (\eps_1)^{\left( \frac{1}{16} \right)^k} \le \min \{ \eps,
  \eta-\norm{M^m}_2\eps \}
\end{align*}
where $\norm{M^m}_2$ is the 2-norm (spectral norm) of the $m^{th}$
moment tensor. We take enough samples so that each estimated moment in
$W$ is within $\min (\eps_1, \eta - \norm{M^m}_2) / n^m)$ of the
Gaussian moment, and every moment in $V$ is off by at most $\min
(\eps_1/2, \eta/2 )$.  In particular, note that all sampled gradients
and Hessian matrices take a value which differ by no more than
$\eps_1/2$ from their true values. Thus, we can simply absorb this as
part of the error arising from local search. Also, this gives us an
upper bound on sample complexity -- the number of samples it takes to
estimate the $m^{th}$ moments of a Gaussian distribution to accuracy
$\eps$ in $\R^n$ is given as $C_m \eps^{-2} n^{m/2} \log n$
\cite{guedon}, which when evaluated becomes $n^{O(m)}$.

At each iteration of the algorithm, we run the Robust Schwartz-Zippel
test $\log (k/\delta)$ times with Schwartz-Zippel parameter
$\eta-\norm{M^m}_2\eps$. With probability at least $1-\delta$, either
each iteration produces a $u$, where $\abs{\E{(x^Tu)^m} - \gamma_m}\ge
\eta-\norm{M^m}_2\eps$ or we correctly deduce that there are no more
directions whose moments differ from a Gaussian by more than $( \eta -
\norm{M^m}_2 \eps) / n^m$. In the latter case, by moment
distinguishability, every vector in $P$, the minimally relevant
subspace, has large projection on the basis $\{ u_i\}$.

In the former case, we know that every unit vector in $\{ u_i\}^\perp$
with projection at least $1-\eps$ takes value which is bounded away
from $\gamma_m$ by at least $\eta - \norm{M^m}_2 \eps$, thus every
such vector is still moment distinguishable. Applying Theorem
\ref{thm:approximate} then, we sequentially generate a sequence of at
most $k$ orthogonal $u_i$ such that:
\begin{align*}
\abs{\innerprod{u_i}{\pi_V(u_i)}}\ge 1-(\eps_1)^{(1/16)^i}
\end{align*}

We need to show that in addition $d_m(F,\hat{F}_U \hat{F}_{U^\perp})
\le \eps$. Let $F'=\hat{F}_U \hat{F}_{U^{\perp}}$: the moment-distance
between the true and sampled distributions differ by at most $\eps_1$,
it suffices for us to prove that $d_m(F,F') \le \eps$. To this end, we
will apply the representation formula to $F'$ for some fixed unit
vector $u$. As before, we have:
\begin{align*}
  \EE{F'}{(x^Tu)^m} &= \EE{F'}{(x^T u_U)^m)} + \EE{F'}{(x^T
    u_{U^\perp})^m}
  - \gamma_m \norm{u_U}^m - \gamma_m \norm{u_{U^\perp}}^m + \gamma_m \\
  & = \EE{F}{(x^T u_U)^m)} + \EE{F}{(x^T u_{U^\perp})^m}
  - \gamma_m \norm{u_U}^m - \gamma_m \norm{u_{U^\perp}}^m + \gamma_m
\end{align*}
Hence, comparing with a similar expression for $\EE{F}{(x^Tu)^m}$:
\begin{align*}
& \abs{\EE{F'}{(x^Tu)^m} - \EE{F}{(x^Tu)^m}}  \le \abs{\EE{F}{(x^T
    u_U)^m)} - \EE{F}{(x^T u_V)^m)}} + \\
& \qquad+ \abs{ \EE{F'}{(x^T u_{U^\perp})^m}-\EE{F}{(x^T
      u_{V^\perp})^m} } \\
& \qquad + \abs{\gamma_m \norm{u_U}^m - \norm{u_V}^m} + \gamma_m \abs{\norm{u_U}^m-\norm{u_{U^\perp}}}
\end{align*}
Now, viewing $\EE{F}{(x^Tu)^m}$ as the tensor applied to $u$, we see
that we can bound these terms by the tensor spectral norm:
\begin{align*}
& \abs{\EE{F'}{(x^Tu)^m} - \EE{F}{(x^Tu)^m}}  \le
(\norm{M^m}_2 + m \gamma_m) \norm{u_U - u_V} + (\norm{M^m}_2 + m \gamma_m)
\norm{u_{U^\perp} - u_{V^\perp}}
\end{align*}
By choice of $U$, we have $\norm{u_U - u_V} \le \eps$, and similarly
for the othogonal component, thus we have our bound.
\end{proof}
}{}

We now apply these methods to learning the concept class $\HH$ (Problem
\ref{problem:learning}). After applying an isotropic transformation, $F$
will have Gaussian moments in every direction orthogonal to $V$, and
hence the output basis of \textbf{FindBasis} and \textbf{ExtendBasis}
returns only vectors in the $V$ subspace. 
\ifthenelse{\boolean{pictures}}{
\begin{algorithm}
  \caption{LearnUnderGaussian}
  \begin{algorithmic}[1]
    \REQUIRE Highest moment $m$, distribution $F$.

    \STATE $B_1 \leftarrow \boldsymbol{FindBasis(m,F)}$.
    \STATE $B_2 \leftarrow \boldsymbol{FindBasis(m,F^+)}$
    on the space orthogonal to $B_1$.
    \STATE Alternately run $\boldsymbol{ExtendBasis}$ on $F$ and $F^+$
    to find a basis $U$ of size at most $k$. Extend this to a basis of
    dimension $k$.
    \STATE Draw sufficient samples $S$ to learn $\HH$ on this $k$
    dimensional subspace. Project $S$ to $\spn{U}$.
    \STATE Learn $\HH$ over $U$.
  \end{algorithmic}
\end{algorithm}
}{}
\ifthenelse{\boolean{longversion}}{
The proof of this algorithm is straightforward given the proof of the
factoring algorithm under Gaussian noise and our robustness
assumptions.
}{}
We will use the following proposition on robust learnability (see e.g., \cite{vempala06}).
\begin{proposition}[VC dimension]\label{proposition:vc}
  Let $\HH$ be a hypothesis class with VC dimension $d$. Let $\ell \in
  \HH$ be a subspace junta with relevant subspace $V$, where
  $dim(V)=k$. Let $U$ be a $k$ dimensional subspace where
  $\ell(\pi_U)$ labels a $1-\eps$ fraction of points correctly. Then
  we can learn $\ell$ with sample complexity $ (1/\eps)^{c_2 d
    \log(1/\eps) + c_2 \log(2/\delta)}$ with probability at
  $1-\delta$.
\end{proposition}

\begin{proof}[Proof of Theorem \ref{thm:learning-gaussian-noise}]  
  $\HH$ is \emph{robust}; by assumption there exists $\eps'$ which is
  polynomial in $\eps$ such that $\eps'+g(\eps') \le \eps/2$. We will
  take this $\eps'$ and will use the following $\eps_1$ for all our
  calls to \textbf{LocalOpt}:
    \begin{align*}
       (\eps_1)^{\left(
          \frac{1}{16} \right)^k} \le \min \{ \eps',
      \eta-\norm{M^m}_2\eps \}
  \end{align*}
  Under these parameters, the proof for the factoring steps of Lines
  1-3 are as in \textbf{FactorUnderGaussian}. Thus with probability at
  least $1-\delta$ we will obtain an orthonormal basis $\{u_i\}$ where
  $\abs{\innerprod{u_i}{\pi_V(u_i)}}\ge 1-(\eps_1)^{(1/16)^i}$.

  By moment learnability, the set of $\{u_i\}$ discovered is
  approximately a basis for $P$, the minimal dimension relevant
  subspace. By our choice of $\eps_1$ above, we satisfy the robustness
  condition, i.e., $\eps_1^{16^k} \le \eps'$, in which
  case only $\eps/2$ fraction of the points are mislabeled over
  $\spn{\{u_i\}}$. Finally, we allow the remaining $\eps/2$ error to
  the learning algorithm, to obtain an output hypothesis which
  correctly labels $1-\eps$ fraction of $F$.
\end{proof}

\section{Moment distance}\label{subsection:fourier}
In our algorithms, we terminate if all remaining directions are
Gaussian in the $m^{th}$ moment (for some fixed $m$). We would like a
guarantee that when we do this, that the random variable is in fact
very close to being Gaussian. What follows is a set of results which
quantify this idea. We first restrict ourselves to one random variable
to introduce the analytic tools we need. In what follows, we use the
following normalisation for our fourier transforms in $\R^n$:
\begin{align*}
  \hat{f}(\xi) = \int e^{i \xi \cdot x} f(x) dx
\end{align*}
This implies that the Parseval/Plancherel theorem takes the following
form:
\begin{align*}
  \int \abs{f(x)}^2 dx = \frac{1}{(2\pi)^n}\int \abs{\hat{f}(\xi)}^2 d\xi
\end{align*}
for $f \in L^2(\R^n)$. 

The core of the proof is the following statement, whose proof employs
Fourier analytic techniques. We need the following standard theorem on
characteristic functions (see for example \cite{shiryaev}):
\begin{theorem}[Characteristic functions]\label{thm:characteristic}
  Let $\xi$ be a random variable with distribution function $F = F(x)$
  and $\phi(t) = \E{ e^{it\xi}}$ its characteristic function. Let
  $\E{\abs{\xi}^n} < \infty$ for some $n \ge 1$, then $\phi^{(r)}$
  exists for all $r \le n$ and
  \begin{align*}
    \phi^{(r)} (t) = \int (ix)^r e^{itx} dF (x)
  \end{align*}
  Moreover, we have an expression for the derivatives at 0:
  \begin{align*}
    \E{\xi^r } = \frac{\phi^{(r)} (0)}{ i^r }
  \end{align*}
  And finally we have the following Taylor series estimate with error:
  \begin{align*}
    \phi(t) = \sum_{r=0}^n \frac{(it)^r}{r!} \E{\xi^r} +
    \frac{(it)^n}{n!} \eps_n (t)
  \end{align*}
  where the error term $\eps_n(t) \to 0$ as $n \to \infty$ and is bounded:
  \begin{align*}
    \abs{\eps_n(t)} \le 3 \E{\abs{\xi}^n}
  \end{align*}
\end{theorem}
Now:
\begin{lemma}[$L^2$ distance from a Gaussian]\label{lemma:L2gaussian}
  Let $f \in L^2 (\R)$ be a probability density over $\R$ whose first
  $m$ moments match those of a standard Gaussian (whose probability
  density we will denote $g$). Suppose that the Fourier transform
  $\hat{f}$ satisifies a tail bound that $\abs{\hat{f}(\xi)} < c /
  \abs{\xi}$ for some $c>0$, then:
  \begin{align*}
    \int_{\R} \abs{f(x)-g(x)}^2 dx \le \frac{c'}{m^{1/8}}
  \end{align*}
\end{lemma}
\begin{proof}
We will assume for the sake of simplicity that $m$ is even. By
Parseval's formula, we have:
\begin{align*}
\int \abs{f(x)-g(x)}^2 dx = \frac{1}{\sqrt{2\pi}} \int \abs{(f-g)\hat{}(\xi)}^2 d\xi
\end{align*}
Both $f$ and $g$ have tail bounds: $f$ by hypothesis, and $g$ because
the Fourier transform of a Gaussian is still a Gaussian. Thus if we
truncate the tails in an interval $[-L,L]$ where $L =m^{1/8}$:
\begin{align*}
  \int_{\R / [-L,L]} \abs{\hat{f}(\xi)}^2 d\xi & \le 2
  \int_{L}^{\infty} \frac{1}{\xi^2} d\xi \\
  & \le \frac{4}{L}
\end{align*}
The Fourier transform of a Gaussian is a Gaussian, and by applying a
standard Gaussian tail bound \cite{Feller}:
\begin{align*}
  \frac{1}{\sqrt{2\pi}} \int_{x}^{\infty} e^{-t^{2}/2} dt \le \left(
    \frac{1}{x} \right) \frac{e^{-x^{2}/2}}{\sqrt{2\pi}}
\end{align*}
We can then combine these estimates using the triangle inequality:
\begin{align*}
  \int_{\R / [-L,L]} \abs{\hat{f-g}(\xi)}^2 d\xi & \le  \int_{\R /
    [-L,L]} \abs{\hat{f}(\xi)}^2 + \abs{\hat{g}(\xi)}^2d\xi \\
  & \le \frac{6}{L}
\end{align*}
In the interval $[-L,L]$, we now apply Theorem \ref{thm:characteristic}:
\begin{align*}
  (\hat{f}-\hat{g})(\xi) & = \sum_{k=0}^m
  \frac{\EE{f}{x^k}-\EE{g}{x^k}}{k!} (i\xi)^k + (\eps_f(t) -
  \eps_g(t))\frac{(i\xi)^m}{m!} \\
  & = (\eps_f(t) - \eps_g(t))\frac{(i\xi)^m}{m!}
\end{align*}
Now we can bound the integral:
\begin{align*}
  \int_{-L}^L \abs{(f - g)\hat{}(\xi)}^2 d\xi & \le \int_{-L}^{L}
  \abs{(\eps_f(\xi) - \eps_g(\xi))\frac{(i\xi)^m}{m!}}^2 d\xi \\
  & \le 6\left( \frac{\E{x^m}}{m!}\right)^2 \int_{-L}^{L} t^{2m} dt \\
  & \le \frac{6}{(2^{m/2}(m/2)!)^2}\frac{2L^{2m+1}}{2m+1} \\
  & \le \frac{12}{2m+1} \exp \left( (2m+1)\log ( L ) - m \log (2)
  - m \log \left( \frac{m}{2} \right) +m \right) \\
  & \le \frac{c}{m} e^{-m}
\end{align*}
\end{proof}
We can also give an approximate version of this theorem:
\begin{lemma}[Approximate moments]\label{lemma:L2approximate}
  Fix $\eps>0$, let $f \in L^2 (\R)$ be a probability density over
  $\R$ whose first $m$ moments satisfy:
  \begin{align*}
    \abs{\EE{f}{x^k} - \gamma_k} \le \eps
  \end{align*}
  Suppose that the Fourier transform $\hat{f}$ satisifies a tail bound
  that $\abs{\hat{f}(\xi)} < c / \abs{\xi}$ for some $c>0$, then (for
  a standard Gaussian $g$);
  \begin{align*}
    \int_{\R} \abs{f(x)-g(x)}^2 dx \le \frac{c'}{m^{1/8}} + c'' m^2
    \eps^2 e^m
  \end{align*}
\end{lemma}
\begin{proof}
  We proceed as in the previous lemma. It suffices for us to bound
  the integral over the interval $[-L,L]$. We apply Cauchy-Schwarz for
  a termwise estimate.
\begin{align*}
  & \int_{-L}^{L} \abs{ \sum_{k=0}^m
  \frac{\EE{f}{x^k}-\EE{g}{x^k}}{k!} (i\xi)^k + (\eps_f(t) -
  \eps_g(t))\frac{(i\xi)^m}{m!}}^2 d\xi \\
&  \le m \int_{-L}^L \sum_{k=0}^m \left(\frac{\EE{f}{x^k}-\EE{g}{x^k}}{k!}
  \xi^k\right)^2 + \left((\eps_f(t)
  -\eps_g(t))\frac{\xi^m}{m!}\right)^2 d\xi
\end{align*}
We can now partition the moments into powers of 2, so consider the
moments where $k \in [m/2^{i+2}, m/2^{i}]$: the integral of each
contributing term is now:
\begin{align*}
\int_{-L}^L&  \left(\frac{\EE{f}{x^k}-\EE{g}{x^k}}{k!}
  \xi^k\right)^2 d\xi = \frac{2(\EE{f}{x^k}-\EE{g}{x^k})^2L^{2k+1}
}{(2k+1)k!} \\
& \le 2\left(\EE{f}{x^k}-\EE{g}{x^k}\right)^2 \exp \left( \frac{2k+1}{8}\log(m)-k \log k
  +k \right)\\
& \le 2\left(\EE{f}{x^k}-\EE{g}{x^k}\right)^2
\exp \left(\frac{(m/2^i)+2}{4}\log(m)-\frac{m}{2^{i+2}} \log \left(
  \frac{m}{2^{i+2}} \right) +
    \frac{m}{2^{i}} \right) \\
& \le 2\left(\EE{f}{x^k}-\EE{g}{x^k}\right)^2
\exp \left( \frac{m}{2^{i-1}} \right)
\end{align*}
\end{proof}

Both of our lemmas so far in this section use a tail bound for the
Fourier transform. One way to obtain such a tail-bound is to examine
logconcave probability densities:
\begin{lemma}[Log-concave densities]\label{lemma:logconcavetailbound}
  Let $f:\R \to \R$ be a logconcave density which is isotropic and
  differentiable, then $\abs{\hat{f}(\xi)} \le 2/\abs{\xi}$.
\end{lemma}
\begin{proof}
  We start by bounding the magnitude of the Fourier transform by the
  integral of the derivative. 
  \begin{align*}
    \hat{f}(\xi) &= \int_\R e^{i \xi x} f(x) dx \\
    & = \int_\R \frac{1}{i\xi}
    \frac{d}{dx} e^{i \xi x} f(x) dx \\
    & = \int_\R \frac{1}{i \xi} e^{i \xi x} \frac{df(x)}{dx} dx
  \end{align*}
where the third line follows by integration by parts and noting that
in the limit $f(x) \to 0$ as $x \to \pm \infty$. This allows us to
bound $\hat{f}(\xi)$:
\begin{align*}
  \abs{\hat{f}(\xi)} \le \frac{1}{\abs{\xi}} \int_\R \abs{f'(x)} dx
\end{align*}
Let us now turn to logconcave densities. Since $f$ is logconcave, we
can write it as $f(x)=e^{h(x)}$ where $h$ is concave. Because $f$ is a
probability density, we must have $h(x) \to -\infty$ as $x \to \pm
\infty$, in which case since $h$ is concave there exists a unique
interval $[a,b]$ where $h(x)$ takes a maximum. This fully determintes
the sign of the derivative: $h'(x) = 0$ in this interval $h'(x) < 0$
for $x < a$ and $h'(x) > 0$ for $x >b $. The same signs pattern holds
for $f'$, as multiplication by $e^{-h(x)}$ does not change the
sign. We can now compute the integral by applying the fundamental
theorem of calculus:
\begin{align*}
  \int_R \abs{f'(x)} dx & = \int_{-\infty}^{a} f'(x) dx + \int_{a}^{b}
  f'(x) dx + \int_{b}^{\infty} -f'(x) dx \\
  & = \lim_{t \to \infty} \left( f(a) - f(-t)\right) +\left( f(b) -
    f(a) \right) + \left( - f(t) + f(b) \right) \\
  & = f(a) + f(b) \\
  & = 2f(a)
\end{align*}
We now apply the following lemma \cite{LV07}, which yields the desired
result.
\begin{lemma}[Upper bound on logconcave functions]\label{lemma:logconcavemax}
  Let $f$ be an isotropic logconcave density in one dimension, then
  $\abs{f(x)} \le 1$.
\end{lemma}
\end{proof}
Then as a corollary to Lemma \ref{lemma:L2gaussian}:
\begin{corollary}[$L^2$ distance for logconcave densities]
\label{corollary:moments}
  Let $f:\R \to \R$ be an isotropic logconcave density whose first
  $m$ moments match a Gaussian $g$, then:
  \begin{align*}
    \norm{f-g} \le \frac{c}{m^{1/8}}
  \end{align*}
\end{corollary}
\begin{proof}
  First, consider the case when $f(x)$ is differentiable. We already
  know that $f \in L^1(\R)$; since $f(x)$ is bounded by 1 (Lemma
  \ref{lemma:logconcavemax}), then we have that $f(x) \in L^2 (\R)$
  because $f(x)^2 \le \abs{f(x)}$. We can now apply Theorem
  \ref{lemma:L2gaussian} with the tail bound guaranteed by Lemma
  \ref{lemma:logconcavetailbound}.

  For the case when $f(x)$ is \emph{not} differentiable, we can
  perturb by a small Gaussian random variable: let $X \sim f$, and let
  $Z \sim N(0,1)$ be an independent normal variable. Fix a parameter
  $\tau \in [0,1]$:
  \begin{align*}
    Y_\tau = (1-\tau) X + \sqrt{2\tau+\tau^2} Z
  \end{align*}
  is isotropic. Moreover, since this the sum of two independent
  logconcave random variables, its density is also logconcave. Let
  $h_1$ denote the density of $(1-\tau)X$ and $h_2$ the density of
  $\sqrt{2\tau+\tau^2} Z$, then the density of our new random variable
  is given by:
  \begin{align*}
    h_1 \ast h_2(x) = \int_{-\infty}^{\infty} h_1(x-t) h_2(t) dt
  \end{align*}
  The convolution of these two distributions is (infinitely)
  differentiable because $h_2$ is (infinitely) differentiable:
  \begin{align*}
    \frac{d}{dx} \left(  h_1 \ast \right)= \left( \frac{d}{dx} h_1
    \right)\ast h_2 = h_1 \ast \left( \frac{d}{dx} h_2 \right)
  \end{align*}
  Thus $Y_\tau$ satisfies the hypotheses of Lemma
  \ref{lemma:logconcavetailbound}, and we have a tail bound for
  $Y_\tau$ as long as $\tau > 0$.

  The first $m$ moments of $Y$ are also close to those of $X$: if we
  compute the $j^{th}$ moment for example:
  \begin{align*}
    \E{ Y_\tau^j } & = \E{\left((1-\tau) X + \sqrt{2\tau+\tau^2} Z\right)^j} \\
    & = (1-\tau)^j \E{X^j} + \sum_{i=1}^j \binom{i}{j} (1-\tau)^j 
    (\sqrt{2\tau+\tau^2})^{i-j} \E{X^i} \E{Z^{i-j}}
  \end{align*}
  Thus we can pick $\tau$ small enough so that:
  \begin{align*}
    \abs{ \E{ Y_\tau^j } - \E{X^j}} \le \eps 
  \end{align*}
  for any $\eps > 0$. In the proof of Lemma \ref{lemma:L2gaussian}
  then, instead of the moment differences from the first $m$ terms of
  the characteristic function being 0, we can make them arbitrarily
  small by choosing smaller $\tau$. Thus we have the conclusion of
  Lemma \ref{lemma:L2gaussian} for $Z$. To conclude, we note that:
  \begin{align*}
    \lim_{\tau \to 0} \norm{h_1 \ast h_2 - f}_2 = 0
  \end{align*}
  in which case, taking $\tau$ small enough allows us to apply the
  triangle inequality to:
  \begin{align*}
    \norm{f-g} \le \norm{f-h} + \norm{h-g}
  \end{align*}
\end{proof}

We also need a lemma to convert our $L^2$ estimates to $L^1$
estimates. This is not general in possible, but since logconcave
functions have exponential tailbounds:
\begin{lemma}[$L^2$ to $L^1$]\label{lemma:L2toL1}
  Let $f,g:\R \to \R$ isotropic logconcave densities such that for
  some $m > 0$ that:
  \begin{align*}
    \int \abs{f(x)-g(x)}^2 dx \le \frac{1}{m}
  \end{align*}
  then:
  \begin{align*}
    \int \abs{f(x)-g(x)} dx \le \frac{c \log (m)}{\sqrt{m}}
  \end{align*}
  for some absolute constant $c > 0$.
\end{lemma}
\begin{proof}
  Fix $L = (\frac{1}{c})\log( m)$, then as before:
  \begin{align*}
    \int \abs{f(x)-g(x)} dx & =     \int_{\abs{
        x} \le L} \abs{f(x)-g(x)} dx  +
      \int_{\abs{x} > L} \abs{f(x)-g(x)} dx
  \end{align*}
  We can now use tail bound for logconcave functions over the tail \cite{gm11}, in
  particular, for isotropic logconcave random variables $X$ in
  $\R^n$, we have (for some fixed absolute constants $c,C > 0$:
  \begin{align*}
    \prob{ \abs{\norm{x}-\sqrt{n}} \ge t \sqrt{n}} \le C \exp \left(
      -cn^\frac{1}{2} \min (t,t^3) \right)
  \end{align*}
  In one dimension, this shows that the integral of our tail is
  bounded by $C / m$ (after application of triangle
  inequality). Now inside the interval $[-L,L]$, we will apply the
  Cauchy-Schwartz inequality:
\begin{align*}
  \int_{[-L,L]} \abs{f(x) - g(x) } dx & \le \left( \int_{[-L,L]}
    \abs{f(x)-g(x)}^2 dx \right)^{1/2} \left( \int_{[-L,L]} 1 dx
  \right)^{1/2} \\
  & \le \frac{\sqrt{2}}{c\sqrt{m}} \log (m) 
\end{align*}
\end{proof}
\begin{proof}[Proof of Theorem \ref{theorem:L1gaussian}]
The proof follows from Lemma \ref{lemma:L2approximate}, Corollary
\ref{corollary:moments} and Lemma \ref{lemma:L2toL1}, noting that the
the technique of Corollary \ref{corollary:moments} can be applied to
Lemma  \ref{lemma:L2approximate} in the same way as Lemma 
\end{proof}

\section{Applications}\label{sec:applications}
In this section, we give some applications of our general theorems and we
some explicit bounds for moment-learnable triples and the
running time of our algorithms on these triples. We make explicit in our
analysis the three key contributions to runtime -- how many moments
are required, how efficiently these moments can be sampled, and how
efficiently the hypothesis can be learned in the $k$-dimensional
relevant subspace.

\subsection{Moment estimation} \label{sec:moments}
In this section, we highlight some further consequences and subtleties
of using moments in algorithms. The use of moments is a very natural
way of studying random variables. For example, the inequalities of
Markov, Chebyshev and Chernoff are statements about the relationship
between a finite sequence of moments and the tail of a
distribution. If we consider an infinite sequence of moments, often
these will determine the distribution uniquely (the moments problem).

One of the critical terms in the runtime given in our main theorems is
$C_F(m, \eps)$: the sample complexity of approximating the $m^{th}$
moment tensor of distribution $F$ to within accuracy (in the moment
metric above). The competitiveness of our algorithm with other
learning algorithms depends on the number of moments we need (ie the
previous section), and the number of samples we need to attain the
required accuracy. This latter problem is well-studied, and there is
an impressive body of literature surrounding it. In particular, when
$m = 2$, the problem is of interest to random matrices community, who
have provided strong bounds in a number of important cases. We will
provide a brief overview of these results, but this by no means is
intended to be a comprehensive survey of the literature!  When the
distribution $F$ is isotropic and almost surely supported in a ball of
radius $O(\sqrt{n})$, Rudelson \cite{rudelson} gave a very strong
bound on $C_F (n, \eps)$ to achieve the following guarantee:
\begin{align*}
\E{\norm{\frac{1}{N}\sum_{i=1}^N x_i x_i^T - I}} \le \eps.
\end{align*}
Rudelson required only $O(n \log(n))$ samples when $F$ is almost
surely supported on a ball of radius $O(\sqrt{n})$, and where the
constant is dependent on $\eps$. Adamczak et al. \cite{adamczak} were
able to improve this bound of $O(n)$ samples. Their assumptions were
support on a ball of radius $O(\sqrt{n})$ as before, and a
subexponential moment condition:
\begin{align*}
\sup_{\norm{v}=1} \E{(x^Tv)^p}^{1/p} = O(p)
\end{align*}
As an application, they showed that logconcave distributions satisfy
these assumptions, and thus their covariance matrices can be sampled
very efficiently. Subsequent work by Vershynin and collaborators
\cite{sriver, vershynin} has broadened the class of efficiently
samplable covariance matrices to distributions where $2 + \eps$
moments exist and also to distributions where the $m^{th}$ moment is
bounded by $K^m$ for some constant $K$.

Finally, in the setting of higher moments, there is the result of
Guedon and Rudelson \cite{guedon}, which gives the sample complexity
of sampling for higher moments of logconcave distributions.  Their
result is that $O(n^{m/2} \log(n))$ samples are necessary to
approximate moments in all directions up to an $1 + \eps$ factor. In
particular, this leads to the observation that explicitly computing a
sample moment tensor from $n^{m/2}$ samples is actually less efficient
than simply storing the points, computing the inner products to the
appropriate powers and summing. This last result is used in our
applications in Section \ref{sec:applications}, as it allows us to
handle many distributions efficiently, including Gaussians and uniform
distributions over convex bodies.

\subsection{Robust learning}\label{sec:learning}

For learning over a $k$-dimensional subspace, we have the following
proposition:
\begin{proposition}[VC dimension]\label{proposition:vc}
  Let $\HH$ be a hypothesis class with VC dimension $d$. Let $\ell \in
  \HH$ be a subspace junta with relevant subspace $V$, where
  $dim(V)=k$. Let $U$ be a $k$ dimensional subspace where
  $\ell(\pi_U)$ labels a $1-\eps$ fraction of points correctly. Then
  we can learn $\ell$ with sample complexity $ (1/\eps)^{c_2 d
    \log(1/\eps) + c_2 \log(2/\delta)}$ with probability at
  $1-\delta$.
\end{proposition}
\begin{proof} To come up with a hypothesis over $U$, we take a new set
  of samples $S$ of size $m$ and project them onto $U$.  By robustness
  of $\HH$ under $F$, we know that $\prob{\ell(\pi_U(x)) =\ell(x)} \ge
  1-\eps$. Then we guess the correct labels by trying all relabelings
  of subsets of size $\eps m$. One of these relabelings will give us a
  labeling consistent with $\ell$ viewed as a function of the
  $k$-coordinates in $U$. For each relabeling we attempt to learn the
  labeling function. On the correct relabeling, we can learn $\ell$ to
  with at most $\eps$ fraction of errors. By the theorem above, our
  total error over $\R^n$ is $2\eps$.

To bound $m$, we apply an idea from \cite{behm89} via a slight
extension (Theorem 5 of \cite{vempala06}). The required bound is $m
\ge (32/\eps) \log(C[m]) + (32/\eps) \log (2/\delta)$ where $C[m]$ is
the maximum number of distinct labelings obtainable using concepts in
$\HH$ over $\R^k$. In particular, we have $C[m] \le \sum_{i=0}^d
\binom{m}{i}$, whence $C[m] \le m^d$. A computation reveals that $m
\ge c (d/\eps) \log(1/\eps) + (c'/\eps) \log( 2/\delta)$ suffices.
The number of relabelings is $\binom{m}{\eps m}$, which is upper
bounded by $(m/\eps)^{m\eps} \le (1/\eps)^{c_2 d \log(1/\eps) + c_2
  \log(2/\delta)}$.
\end{proof}

\ifthenelse{\boolean{longversion}}{
}{}

As mentioned previously, we can view the work of \cite{Vempala10} as a
specialization of our algorithms to the $j=m=2$ case in
\textbf{FindBasis}. We give examples here where the second moment does
\emph{not} suffice, and we must use higher moments to resolve the
relevant subspace $V$. Our examples are:
\ifthenelse{\boolean{longversion}}{ (\ref{example:rectangle})
  hyperrectangles (cuboids) in balls, }{} (\ref{ex:ball}) subsets of
balls, and (\ref{ex:compact}) concepts which have compact support. In
all our examples, the algorithm used is \textbf{LearnUnderGaussian}.
We will prove that we can find the relevant subspaces by running
\textbf{FindBasis} on either the full distribution or distribution
conditioned on positive labels (the ``positive'' distribution).

\ifthenelse{\boolean{longversion}}{
We use the uniform distribution over a ball in $\R^k$ in the relevant
subspace. We need the following elementary fact.
\begin{claim}[Isotropic balls]\label{claim:balls}
  Let $F$ be the uniform distribution (with density $\rho$) over
  $\B_R(0) \subset \R^n$ where $R=\sqrt{n+2}$, then $\E{(x^Tu)^2}=1$
  for any unit vector $u$.
\end{claim}
By a hyperrectangle, we refer to a region of space which is the
Cartesian product of closed intervals i.e.\ $S= [a_i,b_i]\times\cdots
\times [a_k,b_k] \subset \R^k$:
\begin{example}[Hyperrectangles in balls]\label{example:rectangle}
  Let $F=F_VF_W$ where $F_V$ is a uniform distribution over a ball
  $\mathbb{B}$ and $k=dim(V)$, $F_W$ is any Gaussian over $n-k$
  dimensions. Let $S\subset \mathbb{B}$ denote a (hyper)rectangle in
  $V$. Take the hypothesis class $\HH=\{(\chi_S(\pi_V))(x): S \subset
  \mathbb{B}\}$ to be the set of functions which assigns positive
  labels to points whose projection to $V$ lies in the interior of
  rectangle $S$. 
\end{example}

\begin{proposition}
  The triple $(k,F,\HH)$ as defined in Application \ref{example:rectangle}
  is $(4,6/(5k))$ moment-learnable with time and sample complexity
  $\poly (k,1/\eps) + C_{k,\eps} n^2$.
\end{proposition}
\begin{proof}
  Without loss of generality, we may assume that $\mathbb{B} =
  \B_{\sqrt{n+2}}(0)$ after isotropic transformation, and that the
  Gaussian over $F_W$ is a standard $n$-dimensional
  Gaussian. Furthermore, we may assume that $S$ is centered on the
  origin as well (i.e.\ we apply Lemma \ref{lemma:translation} to the
  positively labeled points).

Suppose we now run \textbf{LearnUnderGaussian} on the positively
labeled samples. We start with the second moment ($r=2$) in our
algorithm \textbf{FindBasis}: the second moments of a uniform
distribution over a rectangle are fully determined by the second
moments along the axes of the rectangle. In particular,
\textbf{FindBasis} using the second moments will simply give us every
axis of the rectangle where the second moment is not 1. A simple
calculation of the moments of a uniform distribution over a rectangle
along axis $x_i$ where the rectangle has length $2S_i$ gives:
\begin{align*}
  \E{x_i^2} & = \int_{-S_i}^{S_i} x_i^2 \frac{1}{2S_i} dx_i = \frac{S_i^2}{3}.
\end{align*}
Thus, using the second moment will give us all the axes of our
hyperrectangle except where the rectangle has length $2S_i =
2\sqrt{3}$. Projecting orthogonally to these axes, we now consider the
third moments ($r=3$): the third moment of our uniform rectangle is
clearly 0 in every direction by symmetry of the rectangle. Thus, we
turn to the fourth moment -- note that fixing $S_i=\sqrt{3}$ fixes the
fourth moment along each axis of the rectangle, in particular:
\begin{align*}
  \E{x_i^4} & = \int_{-S_i}^{S_i} x_i^4 \frac{1}{2S_i} dx_i  = \frac{9}{5}.
\end{align*}
Unfortunately, the equality of the fourth moment along the axes of a
rectangle does not necessarily imply the same fourth moment in every
direction. However, iterating Lemma \ref{lemma:representation} allows us to
bound the fourth moments away from the fourth moment of a Gaussian
$\gamma_4=3$:
\begin{align*}
  \E{(x^Tu)^4} = \left(\frac{9}{5} - \gamma_4\right) \sum_{i \in R'}
  u_i^4 + \gamma_4
\end{align*}
where the sum is taken over directions corresponding to axes where
$S_i=\sqrt{3}$. Now by applying the Lagrangian style techniques of
Lemma \ref{lemma:support}, we can bound this by:
\begin{align*}
  \E{(x^Tu)^4} \le \gamma_4 - \frac{6}{5k}
\end{align*}
Thus, we have our moment learnability using only the fourth moment! 
Now that we have the relevant
subspace $V$, we can simply learn our rectangle in a dimension $k$
space, which takes $\poly(k)$ time. Moreover, note that since all the
distributions are logconcave, we can apply the moment sampling results
of Guedeon and Rudelson mentioned in Section \ref{sec:moments} -- in
particular, we can take the number of samples required to be
$C_F(m,\eps) = C_\eps n^2$. Thus this gives a final runtime of
$\poly(k) + C_{k,\eps} n^2$ where $C_{k, \eps}$.
\end{proof}
The key point here is that we have very low polynomial dependence in
$n$. This conforms well with our model where we think of $k$ as being
small compared to $n$. We can, in fact, prove a stronger result --- we
can always find the relevant subspace if $F_V$ is a uniform
distribution over a ball:
}{}
\begin{example}[Uniform distributions over balls]\label{ex:ball}
  Let $F=F_VF_W$ where $F_V$ is a uniform distribution over a ball
  $\B$ and $k=dim(V)$, $F_W$ is a Gaussian.  Let $\HH$ be a robust
  hypothesis class which we can learn with complexity bounded by $T(k,
  \eps)$.
\end{example}
\begin{proposition}
  The triple $(k,F,\HH)$ as defined in Application \ref{ex:ball} is
  $(4,\Omega(1))$ moment-learnable, with the time and sample
  complexity bounded by $T(k,\eps) + C_{k,\eps}n^2$.
\end{proposition}
\ifthenelse{\boolean{longversion}}{
\begin{proof}
  We will examine what happens when we run \textbf{FindBasis} on the
  full distribution (as opposed to the positive distribution in the
  previous example).  We compute the fourth moment of a ball of radius
  $R=\sqrt{n+2}$. For simplicity, we will assume that $k=2l+1$ for
  some positive integer $l$ ie $k$ is odd:
  \begin{align*}
    \E{x_1^4} & = \int_{\B_R(0)} x_1^4 \rho dx \\
    & = \int_{-R}^{R} \int_{\B_{\sqrt{R^2-x_1^2}}^{k-1}(0)} x_1^4 \rho
    dx_2 \cdots dx_k dx_1 \\
    & = \frac{1}{\vol{\B_R^k(0)}} \int_{-R}^{R} x_1^4
    \vol{\B_{\sqrt{R^2-x_1^2}}^{k-1}(0)} dx_1 \\
    & = \frac{\vol{\B_R^{k-1}(0)}}{\vol{\B_R^k(0)}} \int_{-R}^{R} x_1^4
    \left( 1-\frac{x_1^2}{R^2}\right)^l dx_1
  \end{align*}
  We first examine the volume ratio: using the recurrence:
  \begin{align*}
    \vol{\B^{k}_R(0)} = \frac{2\pi R^2}{k} \vol{\B^{k-2}_R(0)}
  \end{align*}
  and unrolling the recurrence, we have:
  \begin{align*}
    \frac{\vol{2l}}{\vol{2l+1}} &=\frac{(2l+1)!!}{2R(2l)!!} \\
    & = \frac{1}{2R} \frac{(2l+2)!!}{(l+1)!2^{l+1} l! 2^l} \\
    & = \frac{1}{2R} \frac{(2l+2)!!}{(l+1)! l! 2^{2l+1}} \\
  \end{align*}
  Applying Stirling's approximation, we have:
\begin{align*}
  \frac{\vol{2l}}{\vol{2l+1}} &= \frac{1}{2R}
  \frac{\sqrt{2\pi(2l+2)}}{2\pi \sqrt{l(l+1)}} \frac{1}{2^{2l+1}}
  \left( \frac{2l+2}{e} \right)^{2l+2}
  \left( \frac{e}{l} \right)^l
  \left( \frac{e}{l+1} \right)^{(l+1)} \\
& = \frac{1}{2R} \frac{1}{\sqrt{\pi l}} \frac{2}{e} \left(
  \frac{l+1}{l} \right)^l (l+1) \\
& = \frac{1}{R\sqrt{\pi}} \frac{l+1}{\sqrt{l}}\\
& = \frac{1}{\sqrt{\pi}} \left( 1 + \sqrt{\frac{1}{l(l+2)}} \right)
\end{align*}
Returning to the integrand, we can simplify it somewhat:
\begin{align*}
\int_{-R}^{R} x_1^4
    \left( 1-\frac{x_1^2}{R^2}\right)^l dx_1 &= 2 \int_{0}^{R} x_1^4
    \left( 1-\frac{x_1^2}{R^2}\right)^l dx_1
\end{align*}
By explicitly taking the integral (using a computer algebra system), we have:
\begin{align*}
\int_{0}^{R} x_1^4
    \left( 1-\frac{x_1^2}{R^2}\right)^l dx_1 = \frac{3 \sqrt{\pi} (2 l+3)^{5/2}
     \Gamma ( l+1)}{8 \Gamma (l+7/2)}
\end{align*}
where $\Gamma$ here is the usual gamma function. The behavior of this
function is as follows:
\begin{align*}
  \lim_{l \to \infty} \frac{3 \sqrt{\pi} (2 l+3)^{5/2}
     \Gamma ( l+1)}{8 \Gamma (l+7/2)} = 3\sqrt{\frac{\pi}{2}}
\end{align*}
Moreover, the function is monotonic increasing for $l > 0$, and takes
on the value $56\sqrt{7}/45$ at $l=2$. Thus, combining these facts
with the estimate of the volume ratios, we can see that the fourth
moment of a ball is bounded away from the fourth moment of a standard
Gaussian by a constant, hence we can take $\eta = \Omega(1)$. Once we
have the relevant subspace $V$, we can project the samples to $V$ and
learn in time $T(k,\eps)$. The runtime in this case is $T(k,\eps)
+C_{k,\eps} n^2$.
\end{proof}
}{} As a specialization, when the positive examples are determined by
a convex subset of the unit ball, $T(k,\eps) \le (k/\eps)^{O(k)}$. In
a $k$-dimensional subspace, we can learn a convex subset of the ball
by simply taking the convex hull of $(k/\eps)^{O(k)}$ random positive
points. From the classical approximation theory of convex bodies
\cite{ConvGeomHandbook}, we obtain an approximation to the true convex
body to within relative error $\eps$, giving total runtime
$(k/\eps)^{O(k)} + C_{k,\eps}n^2$.
\ifthenelse{\boolean{longversion}}{ This complements \cite{Vempala10}
  which provides a PCA-based algorithm for learning convex bodies when
  the distribution in the relevant subspace is also Gaussian.  In that
  paper, it is mentioned that standard PCA fails if the full
  distributions is not a Gaussian.

}{} We now present an example that relies on boundedness -- either of
the full distribution in the relevant subspace, or the positive
distribution. This rather general result uses relatively many moments.

\begin{example}[Compact distribution in relevant subspace]\label{ex:compact}
  Let $F=F_VF_W$ where $F_W$ is any Gaussian over $n-k$
  dimensions. Take $\HH$ to be a robust hypothesis class learnable
  with complexity $T(k,\eps)$. Assume that either $F_V$ or $\HH$ has
  its support contained in $B_{g(k)}(0)$.
\end{example}

\begin{proposition} \label{proposition:compact}
  The triple $(k,F,\HH)$ described in Application \ref{ex:compact} is
  $(g(k),\Omega(1))$ moment-learnable with complexity $T(k,\eps) +
  C_{k,\eps} n^{O(g(k)^2)}$.
\end{proposition}
\ifthenelse{\boolean{longversion}}{
\begin{proof}
  Suppose we run \textbf{FindBasis} on the full distribution or the
  positive distribution, whichever is contained in a ball of radius
  $g(k)$. Consider the relevant subspace. If we fix some even moment
  $m$ then we can give explicit bounds on the moments:
  \begin{align*}
    \E{(x_t)^m} & \le g(k)^m.
  \end{align*}
  On the other hand, the even moments of a Gaussian are given by
  $(m-1)!!=m!/(m/2)!2^{m/2}$ which grows much more rapidly. If we take
  logarithms on both sides, then we can find $m=m(k)$ such that:
  \begin{align*}
    m \log (g(k)) \le \log \left( \frac{m!}{(m/2)!2^{m/2}}
    \right)
  \end{align*}
  Applying Stirling's approximation yields:
  \begin{align*}
        \frac{m}{2}\log (g(k)^2) & \le m \log (m ) - m - \frac{m}{2} \log
        \left( \frac{m}{2} \right) + \frac{m}{2} - \frac{m}{2} \log
        (2) \\
        & \le \frac{m}{2} \log(m) - \frac{m}{2}
  \end{align*}
  So if we pick $m=2g(k)^2$, then the difference in the moments should
  be $\Omega (1)$.  Thus, simply running \textbf{FindBasis} on the
  full distribution will allow us to recover the relevant subspace, at
  which point we can learn $\HH$ in $\R^k$ (doable in time $T(k)$). It
  remains to prove that we can sample the first $2g(k)^2$ moments of a
  bounded distribution efficiently: since it is bounded, all moments
  exist. In particular, if we require $2g(k)^2$ moments, then the
  $4g(k)^2$ moment is bounded by $g(k)^{4g(k)^2}$. Then by applying
  Chebyshev's inequality, we see that we need at most
  $g(k)^{O(g(k)^2)}$ samples in the relevant subspace. The overall
  runtime for this algorithm is then $T(k,\eps)+C_{k,\eps}
  n^{O(g(k)^2)}$.
\end{proof}
}{}

\bibliographystyle{plain}
\bibliography{kplanes}

\end{document}